\documentclass[acmsmall,screen]{acmart}
\renewcommand\footnotetextcopyrightpermission[1]{}

\usepackage[lined,boxed,commentsnumbered, ruled,vlined, resetcount]{algorithm2e}
\usepackage{multirow}
\usepackage{xspace}
\usepackage{enumitem}
\usepackage{amsmath}
\usepackage{hyperref}
\usepackage{microtype}
\usepackage{xiaostyle}
\usepackage{tabularx}
\usepackage{makecell}
\usepackage{subcaption}
\usepackage{wrapfig}
\usepackage{adjustbox}
\usepackage{algpseudocode}
\usepackage{tikz}
\usepackage{svg}
\usepackage{listings}
\usepackage{xcolor}
\usepackage{fontawesome}

\usepackage{graphicx}
\usepackage{booktabs}

\usetikzlibrary{automata, positioning, arrows}
\usetikzlibrary{chains}

\begin{document}

\title{Tracing Errors, Constructing Fixes: Repository-Level Memory Error Repair via Typestate-Guided Context Retrieval}

\author{Xiao Cheng}
\email{jumormt@gmail.com}
\orcid{0000-0001-5456-3827}
\affiliation{%
  \institution{University of New South Wales}
  \city{Sydney}
  \state{NSW}
  \country{Australia}
  \postcode{2052}
}  
\authornote{Xiao Cheng and Zhihao Guo contributed equally to this work.} 

\author{Zhihao Guo}
\email{guodududu@gmail.com}
\orcid{0009-0000-4130-1270}
\affiliation{%
  \institution{University of Technology Sydney}
  \city{Sydney}
  \state{NSW}
  \country{Australia}
  \postcode{2000}
}  
\authornotemark[1]

\author{Huan Huo}
\email{huan.huo@uts.edu.au}
\orcid{0000-0003-2440-714X}
\affiliation{%
  \institution{University of Technology Sydney}
  \city{Sydney}
  \state{NSW}
  \country{Australia}
  \postcode{2000}
}  

\author{Yulei Sui}
\email{y.sui@unsw.edu.au}
\orcid{0000-0002-9510-6574}
\affiliation{%
  \institution{University of New South Wales}
  \city{Sydney}
  \state{NSW}
  \country{Australia}
  \postcode{2052}
}

\renewcommand{\shortauthors}{Xiao Cheng}

\begin{abstract}
Memory-related errors in C programming continue to pose significant challenges in software development, primarily due to the complexities of manual memory management inherent in the language. These errors frequently serve as vectors for severe vulnerabilities, while their repair requires extensive knowledge of program logic and C's memory model. Automated Program Repair (APR) has emerged as a critical research area to address these challenges. Traditional APR approaches rely on expert-designed strategies and predefined templates, which are labor-intensive and constrained by the effectiveness of manual specifications. Deep learning techniques offer a promising alternative by automatically extracting repair patterns, but they require substantial training datasets and often lack interpretability.

This paper introduces \thistool, a novel approach that harnesses the potential of Large Language Models (LLMs) for automated memory error repair, especially for complex repository-level errors that span multiple functions and files. We address two fundamental challenges in LLM-based memory error repair: a limited understanding of interprocedural memory management patterns and context window limitations for repository-wide analysis. Our approach utilizes a finite typestate automaton to guide the tracking of error-propagation paths and context trace, capturing both spatial (memory states) and temporal (execution history) dimensions of error behavior. This typestate-guided context retrieval strategy provides the LLM with concise yet semantically rich information relevant to erroneous memory management, effectively addressing the token limitation of LLMs. 
Our framework has successfully repaired \fixmemerr out of \totalmemerr real-world memory errors
derived from \pronum open-source projects that collectively comprise over a million lines of code. Compared to state-of-the-art memory error APR tools, \SAVER and \ProveNFix, our approach correctly fixes \compares and \comparep more errors, respectively. 
Moreover, \thistool outperforms current open-source state-of-the-art LLM-based SWE-agent 1.0 by repairing 94\% more errors while consuming 17M (41$\times$) less tokens.
We have also successfully repaired three critical zero-day memory errors, with fixes that have been accepted and implemented by the original developers. These results highlight a promising paradigm for repository-level program repair through program analysis-guided, retrieval-augmented LLMs, combining formal verification strengths with neural model adaptability.

\end{abstract}



\maketitle

\section{Introduction\label{sec:introduction}}

Memory-related errors in C programming constitute a persistent and formidable challenge in software development~\cite{Memfix, yan2018uaf, xie2005memleak, sui2012memleak, caballero2012undangle}. These errors frequently serve as vectors for zero-day attacks, resulting in severe consequences such as data corruption~\cite{wang2019locating, xu2015collision}, denial-of-service incidents~\cite{chen2011linux, borisov2007denial}, and information leakage~\cite{song2020information, lee2015preventing}. The intrinsic complexity of memory management in C, coupled with the language's low-level semantics, renders the detection and repair of these errors both labor-intensive and error-prone. Consequently, automated program repair (APR), which mitigates the need for exhaustive manual error analysis and repair, has emerged as a critical research domain in software engineering over the past decade~\cite{history_apr, Gazzola2019Survey, Monperrus2018intro_APR, APR_survey}.

Conventional APR approaches~\cite{SAVER, Footpatch, LeakFix, Memfix, semfix, mechtaev2016angelix, directfix, ProveNFix} rely on predefined repair templates and expert-crafted heuristics. The manual creation of repair rules is labor-intensive, and the effectiveness of repairs is constrained by the precision and completeness of human-specified rules. Moreover, even when memory errors can be successfully replayed and their root causes identified, the repair process can be still challenging without a comprehensive understanding of repository-level context and program semantics~\cite{SAVER}. For instance, addressing memory leaks may require meticulous handling of complex data structure deallocation, while resolving use-after-free errors demands careful consideration of pointer aliasing relationships and strategic control flow restructuring (see examples in Figures~\ref{fig:motivating} and ~\ref{fig:case-study}). These complexities necessitate not only an understanding of program-specific memory management paradigms but also a comprehensive grasp of broader program semantics to ensure overall program correctness rather than merely eliminating the immediate error manifestation.

Recent advancements in deep learning (DL) have emerged as a promising direction to address these limitations through their capability to learn complex program semantics and repair patterns from large-scale codebases~\cite{NMT_CoCoNut20, NMT_Edits20, NMT_knod23, NMT_Recoder21, NMT_SequenceR19, NMT_tare23, NMT_Tufano19, OLLM_Dear22, Li2020DLFix}. In contrast to rule-based approaches, DL methods can automatically extract and generalize fix patterns across diverse contexts, eliminating the need for manually crafted repair templates. These approaches excel at capturing intricate relationships between code structure, program semantics, and memory management patterns, potentially enabling more sophisticated and context-aware repairs.
However, DL-based approaches require substantial amounts of memory error repair data for training, which is particularly challenging to obtain given the relative scarcity and complexity of these errors in real-world codebases. Additionally, the inherent opacity of deep learning models---their "black box" nature---obscures the reasoning behind generated fixes, potentially undermining developer confidence in their effectiveness and making it difficult to validate the correctness of proposed repairs.

Large Language Models (LLMs) have emerged as a compelling alternative to both traditional rule-based and DL-based approaches for automated memory error repair. Unlike conventional methods constrained by predefined templates and expert-crafted rules, LLMs leverage their extensive training on vast code repositories and natural language corpora to comprehend complex program semantics. Furthermore, in contrast to DL approaches that demand substantial memory-error-specific training data, LLMs can utilize their generalized understanding of comprehensive codebases to address repair tasks with limited domain-specific examples. However, despite these inherent advantages, the potential of LLMs for memory error repair remains largely unexplored. While existing LLM-based approaches~\cite{CLLM_chatrepair, CLLM_CLPR, CLLM_CodeRover, CLLM_inferfix, CLLM_SRepair, OLLM_2024_Xia, wei2023LLMAPR} have demonstrated success in general bug-fixing scenarios, they fall short in addressing the unique challenges of memory error repair.

Memory errors frequently manifest as interprocedural phenomena requiring comprehensive semantic understanding across entire repositories---a complexity that exceeds the capabilities of current approaches primarily focused on localized fixes within isolated functions or files.
To illustrate this challenge empirically, consider a use-after-free error (detailed in \S\ref{sec:motivating}). This error manifests across six distinct functions, with a correct repair requiring synchronized modifications at three separate locations. In our experimental evaluation, we found that LLMs, even when provided with both the error-triggering function and its complete calling chain alongside a precise error root cause analysis, consistently failed to generate semantically correct patches. This limitation stems fundamentally from the inherent complexity of repository-level memory error repair, which demands a comprehensive understanding of long-range memory management contexts and interprocedural dependencies. These challenges can be characterized as follows:

\textbf{Challenge\#1: Limited understanding of interprocedural memory management patterns.}
Memory errors in production systems frequently manifest through sophisticated management patterns that span multiple functions, involve intricate data structures, and exhibit complex pointer aliasing relationships. While LLMs have demonstrated remarkable capability in comprehending intraprocedural program structures~\cite{li2024enhancing}, their ability to grasp the nuanced semantics of memory errors---particularly those with complex interprocedural triggering patterns---remains limited.

\textbf{Challenge\#2: Context window limitations for repository-wide analysis.}
Memory errors often span multiple functions and execution contexts across a repository, with both error-triggering logic and corresponding repair patches potentially distributed throughout various functions in large-scale production codebases. This presents a fundamental challenge for LLMs, which are constrained by token limitations and exhibit performance degradation with extensive prompts---a phenomenon known as ``lost in the middle''~\cite{liu2024lost}. These constraints significantly impair LLMs' ability to process and reason about the comprehensive contextual information necessary for effective memory error analysis and repair.

Drawing inspiration from established developer practices for memory error resolution, we propose an approach that methodically emulates this systematic debugging process. When addressing memory errors, developers typically follow three essential steps: 1) reproducing the error through test cases and specialized program analysis tools such as ASan~\cite{ASan}; 2) deploying debugging utilities like GDB~\cite{GDB} to establish strategic breakpoints and comprehensively analyze program dependencies and execution context; and 3) applying domain-specific knowledge of safety specifications to implement semantically robust fixes.
Based on these observations, we present \thistool, a novel LLM-based approach for automatic memory error repair in C programs that leverages context-aware retrieval augmented by typestate analysis~\cite{strom1986typestate}. Our architecture positions the LLM as a reasoning engine that operates synergistically with established program analysis techniques and debugging infrastructure to facilitate comprehensive error comprehension and correction. This integration effectively bridges the gap between LLMs' linguistic capabilities and the specialized contextual understanding required for effective memory error repair. 

To address \textbf{Challenge\#1}, we formalize memory error semantics through tracking the error-propagation path and monitoring program context transitions. The propagation path encodes the temporal execution history leading to error manifestation, while the context at each critical program point encapsulates both memory management states and detailed backtrace information of interprocedural calling chains. This comprehensive representation enables LLMs to model complex state evolution patterns across function boundaries by preserving both spatial (memory states) and temporal (execution history) dimensions of error behavior.
In response to \textbf{Challenge\#2}, we introduce an efficient context collection strategy governed by a finite typestate automaton (FTA)~\cite{fink2006typestate, strom1986typestate, xiao2024fgs, das2002esp, bodden2010typestate}. Rather than exhaustively capturing contexts at every execution point---which would overwhelm LLMs' context windows---the FTA guides the collection process by monitoring object lifetime phases from allocation to error manifestation, triggering context snapshots exclusively at semantically significant state transitions. This selective approach generates a concise yet semantically rich context trace that preserves essential memory management semantics while enabling LLM-based repair to scale effectively to repository-wide analysis.

\begin{figure*}[t]
    \centering
        \includegraphics[width=1.0\textwidth]{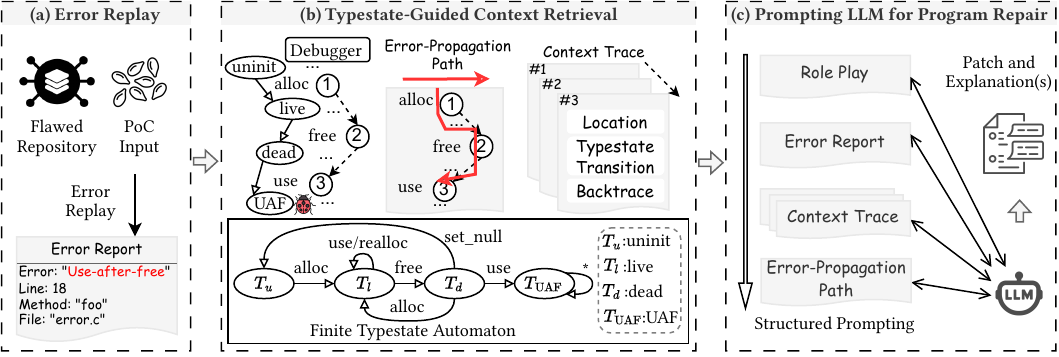}
        \vspace{-8mm}
        \caption{An overview of our framework.}
         \vspace{-8mm}
    \label{fig:framework}
\end{figure*}

Figure~\ref{fig:framework} provides an overview of our framework consisting of three phases:

\textbf{(a) Error Replay.} The input to this framework is a flawed code repository and its proof of concept (PoC) input, which then undergo a preliminary analysis that uses a dynamic analysis tool to replay the memory error and pinpoint the error-triggering location and the specific error type.

\textbf{(b) Typestate-Guided Context Retrieval.} 
Guided by a finite typestate automaton (FTA), we use a debugger to execute the program step-by-step, extracting the error-propagation path from the nearest related memory allocation to the error-triggering point. We then construct a context trace of the memory error, starting from the memory allocation and following each memory operation across the program. The context tracing is guided by the FTA, allowing for efficient tracking of program contexts only at typestate-changing breakpoints. Each context includes three elements: the location of the current point, the typestate transition, and the backtrace of the calling stack gathered at the current breakpoint.

\textbf{(c) Prompting LLM for Program Repair.} Finally, we design a multi-step structured prompting method that incrementally deliver 
role and task description~\cite{shanahan2023role}, error report, context trace and error-propagation path to the LLM for generating an appropriate patch and explanation(s).  


Our major contributions are as follows:
\begin{itemize}[noitemsep, topsep=1pt, partopsep=1pt, listparindent=\parindent, leftmargin=*]
    \item We introduce \thistool, the first repository-level memory error repair system that synergizes Large Language Models (LLMs) with context-aware semantic retrieval and typestate-guided program analysis. Our methodology advances conventional APR by enabling LLMs to systematically derive correct patches through establishing memory error semantics via interaction with runtime debuggers and formal typestate verification.
    \item 
    We propose a novel typestate-guided context retrieval methodology that precisely captures and synthesizes critical memory error contextual information during typestate transitions. This targeted approach enables LLMs to comprehend complex memory error semantics while maintaining focus on concise, semantically relevant code segments, thereby significantly improving repair accuracy and efficiency.
    \item We have established a real-world memory error database encompassing \totalmemerr errors, their proofs of concept (PoCs), and their respective fixes across \pronum projects, collectively containing over a million lines of code. Our framework successfully repairs \fixmemerr memory errors, surpassing the performance of current automated memory error repair tools. We also successfully addressed three zero-day memory errors, with our solutions accepted and implemented by the original developers. Full details will be disclosed upon paper acceptance.
\end{itemize}



\section{Preliminaries and Problem Formulation\label{sec:background}}

This section establishes the theoretical foundation for memory error detection using typestate analysis and examines the capabilities of Large Language Models (LLMs) in code comprehension. We then formulate the automated program repair problem for memory errors within the context of LLM-based approaches.

\subsection{Memory Errors and Typestate Analysis}


Memory errors in C programming often stem from improper memory management, necessitating careful tracking of memory's temporal properties along the program's control flow. Typestate analysis~\cite{das2002esp, strom1986typestate, fink2006typestate, bodden2010typestate} emerges as an effective method for detecting and understanding these errors, as it monitors execution logic by tracking the temporal state changes of memory objects. This approach represents different states of a given memory object and their transitions using a finite typestate automaton (Definition~\ref{def:tfsa}), allowing for precise modeling of memory object lifecycles. By capturing the various states a memory object can occupy (such as allocated, initialized, or freed) and tracking transitions between these states during program execution, typestate analysis can identify violations of expected state sequences---often indicative of memory errors.


\begin{definition}[Finite Typestate Automaton]
\label{def:tfsa}
A finite typestate automaton (FTA) for an error ET is a quintuple denoted as \(\mathcal{A}_{\textup{ET}}=\left<\Sigma ,\mathbb{T},T_u,\delta,T_\textup{ET} \right>\). 
The language \(\Sigma\) signifies the operations (e.g., function calls) that can be performed on the typestates. \(\mathbb{T}\) encompasses all the possible typestates, with \(T_u \in \mathbb{T}\) representing the initial state. \(\delta: (\mathbb{T}\times\Sigma)\to\mathbb{T}\) is the state-transition table encoding the effects of operations in \(\Sigma\). \(T_\textup{ET}\) is the error typestate indicating a potential error detected. For a program statement \(s\), we use \(\mathtt{op}(s)\) to retrieve its corresponding operation in \(\Sigma\).
\end{definition}

\begin{table}
    \footnotesize
    \renewcommand{\arraystretch}{1.2} 
    \caption{Finite typestate automata of use-after-frees (UAF), double-frees (DF) and memory leaks (ML).}
    
    \vspace{-3mm}
    \begin{adjustbox}{width=1\textwidth, center}
    \begin{tabular}{|l|l|l|}
    \hline
        $\mathcal{A}_{\text{UAF}} = \left<\Sigma, \mathbb{T}, T_u, \delta, T_\textup{UAF}\right>$ & 
        $\mathcal{A}_{\text{DF}} = \left<\Sigma, \mathbb{T}, T_u, \delta, T_\textup{DF}\right>$ & 
        $\mathcal{A}_{\text{ML}} = \left<\Sigma, \mathbb{T}, T_u, \delta, T_\textup{ML}\right>$ \\ 
    \hline
        $\mathbb{T} = \{T_u, T_l, T_d, T_\textup{UAF}\}$ & 
        $\mathbb{T} = \{T_u, T_l, T_d, T_\textup{DF}\}$ & 
        $\mathbb{T} = \{T_u, T_l, T_d, T_\textup{ML}\}$ \\ 
        $\Sigma = \{\text{alloc}, \text{free}, \text{use}, \text{realloc}, \text{set\_null}\}$ & 
        $\Sigma = \{\text{alloc}, \text{free}, \text{realloc}, \text{set\_null}\}$ & 
        $\Sigma = \{\text{alloc}, \text{free}, \text{realloc}, \text{exit}\}$ \\ 

    \hline
        \includegraphics[scale=0.78]{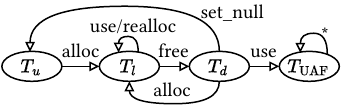} &
        \includegraphics[scale=0.78]{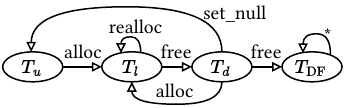} &
        \includegraphics[scale=0.78]{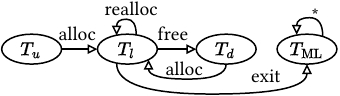} \\ \hline
$T_\textup{UAF}$: Use-after-free error state.& $T_\textup{DF}$: Double-free error state.&$T_\textup{ML}$: Memory leak error state.
        \\\hline
    \multicolumn{3}{|l|}{$T_u$: Uninitialized state (memory object is not yet allocated); $T_l$: Live state (memory object is allocated and in use); }
    \\
    \multicolumn{3}{|l|}{
    $T_d$: Dead state (memory object is released).
    }\\\hline
    \multicolumn{2}{|l|}{
        use: Use a heap object; {set\_null}: Set the pointer pointing the heap object to null.
        } & exit: Return from main function
        \\\hline
    \multicolumn{3}{|l|}{
    {alloc}: Allocate a heap memory object; {realloc}: Reallocate a heap memory object; {free}: Free a heap object;  
    }
    \\\hline
    \end{tabular}
    \end{adjustbox}
    \label{tab:fsm}
\vspace{-3mm}
\end{table}

\runinhead*{Specifications for Memory Errors.} 
In this paper, we focus on three critical yet difficult-to-fix memory errors~\cite{SAVER}: use-after-frees~\cite{CWE416}, double-frees~\cite{CWE415}, and memory leaks~\cite{CWE401}. Table~\ref{tab:fsm} presents the specifications of these errors in the form of finite typestate automata. Each automaton is represented as a graph where each node corresponds to a specific typestate, and the edges between these nodes are annotated with transition operations, illustrating the transition relationships between different states.
The analysis process begins with an uninitialized typestate, denoted as \(T_u\), and then may advance through different typestate transitions depending on the program statements encountered. For example, \(T_u\) transitions to \(T_l\) upon encountering a memory allocation statement, such as the primitive heap allocation API \texttt{malloc}. If released memory (\(T_d\)) is used or freed again, it transitions to \(T_\textup{UAF}\) and \(T_\textup{DF}\), representing a use-after-free or double-free error as per \(\mathcal{A}_{\text{UAF}}\) and \(\mathcal{A}_{\text{DF}}\) respectively. Similarly, a live heap memory object (\(T_l\)) transitioning to \(T_\textup{ML}\) at program exit as per \(\mathcal{A}_{\text{ML}}\) indicates a memory leak.

\subsection{Large Language Models for Code Comprehension}


Large Language Models (LLMs) demonstrate exceptional proficiency in code comprehension and manipulation, establishing themselves as powerful tools for software engineering tasks. Trained on extensive corpora encompassing both natural language and source code, these models exhibit sophisticated capabilities in parsing and interpreting complex program structures.
Recent research has shown that LLMs excel in several domains critical for program analysis: they effectively interface with external APIs~\cite{wang2024llmapi}, resolve indirect call relationships~\cite{cheng2024semanticenhancedindirectanalysislarge}, and conduct nuanced data flow analyses~\cite{wang-etal-2024-sanitizing,wang2025llmdfa}. Their capacity extends to recognizing intricate program dependencies, evaluating path conditions, and inferring control flow constructs~\cite{huang2024revealing}, which enables robust reasoning about diverse execution paths, including those within iterative structures~\cite{li2024enhancing}. Moreover, these models demonstrate significant utility in understanding and synthesizing formal verification components, such as loop invariants~\cite{pei2023can} and contractual specifications including preconditions and postconditions~\cite{wen2024enchanting}. These advanced code comprehension capabilities of LLMs present compelling opportunities for their application in automated memory error repair---the primary focus of our research.

\subsection{Problem Formulation}

Our objective is to automatically generate a patch to repair a flawed C repository with memory errors by leveraging the power of LLM and the prompts constructed based on FTA specifications (Table~\ref{tab:fsm}). The patch should fix the memory error and not introduce new bugs. 
Formally, let \( P \) be the original flawed C repository (the error type is ET) and \( I \) be a specific proof of concept (PoC) input. We generate a set of prompts \( \llbracket Q \rrbracket \) via the function \textbf{M}:
\[
\llbracket Q \rrbracket = \text{\textbf{M}}(P, \mathcal{A}_\textup{ET}, I)
\]

The LLM's role is to take \( \llbracket Q \rrbracket \) to generate a correct patch \( \Delta P \):
\[
\Delta P = \text{LLM}(\llbracket Q \rrbracket)
\]

The patch \( \Delta P \), when applied to \( P \), should yield a updated repository \( P'=P+\Delta P \) that fixes the error without introducing new bugs. 

To ensure the patch's effectiveness and reliability, we impose these constraints:

(1) The patched repository \( P' \) should not exhibit the memory error when executed with \( I \):
\[
f(P', I) \neq \text{memory error}
\]

where \( f(P, I) \) denotes the execution of \( P \) with input \( I \).

(2) Let \( \mathcal{I} \) represent the available test-suite for the repository. The patched repository \( P'\) should not introduce new bugs, meaning that for any test case in the provided test-suite \( I' \in \mathcal{I} \), the execution \( f(P', I') \) should produce correct results:
\[
\forall I' \in \mathcal{I}, \ f(P', I') = \text{expected result}
\]


We formulate our memory error APR problem as follows:
\begin{formulateBox}
Given a flawed C repository \( P \) and a PoC input \( I \), we design a method \( \text{\textbf{M}} \) to construct prompts \( \llbracket Q \rrbracket \) that guide the LLM to generate a patch \( \Delta P \). The goal is to ensure that the patch \( \Delta P \), when applied to \( P \), fixes the memory error for \( I \) and maintains correctness for its test-suite \( \mathcal{I} \).  
\end{formulateBox}

\section{A Motivating Example\label{sec:motivating}}

Figure~\ref{fig:motivating} illustrates the pipeline of \thistool, walking through the three phases depicted in Figure~\ref{fig:framework}. These phases are demonstrated using a use-after-free error.
A heap memory object is initially allocated by the \code{create\_context} function, which wraps a \code{malloc} call at Line 27 in the \texttt{test.c} file (\textcircled{1}). 
This memory is freed by invoking the \code{release\_context} method, which calls a \code{free} function at Line~37 in \texttt{test.c} (\textcircled{2}). 
However, the released memory is erroneously used in the method \code{clone\_data} at Line~21 in \texttt{test.c} (\textcircled{3}).
This use-after-free error traverses six functions from allocation (\textcircled{1}) to the error-triggering point (\textcircled{3}).

\textbf{Challenges.} The successful repair of this memory error poses four fundamental challenges.
First, it demands a thorough understanding of the allocated memory structure (\code{Context}), necessitating proper deallocation and null pointer validation mechanisms for both the base object (\code{Context}) and its associated fields (\code{data}).
Second, the repair requires careful restructuring of operation sequences to maintain program correctness, particularly ensuring that \code{release\_context} executes after its dependent operation \code{copy\_ctx}, thus preventing critical logic in \code{copy\_ctx} from being invalidated by premature null pointer checks.
Third, the fix requires semantic comprehension of the codebase to implement a deep copy operation using \code{memcpy} for \code{src}, rather than a potentially unsafe shallow copy of \code{src->data}.
Finally, the repair must correctly handle interprocedural interactions along the error-propagation path, ensuring consistent memory management across function boundaries.

\textbf{Existing Efforts.} Notably, state-of-the-art memory error APR tools, including \ProveNFix~\cite{ProveNFix} and \SAVER~\cite{SAVER}, fail to repair this memory error even when being provided with \emph{the precise error location and configured with their most sophisticated options}, such as flow-sensitive analysis and header file parsing. They fall short in understanding the hierarchical structure of the \code{Context} object, consequently implementing only superficial null checks on \code{src} rather than comprehensively addressing the whole structure. Their repair strategies further exhibit semantic miscomprehension by introducing premature returns after the \code{release\_context} call which, while eliminating the immediate error, prevents the execution of subsequent critical operations and compromises program integrity. The LLM-based approach SWE-agent 1.0~\cite{yang2024sweagent} similarly fails to generate an appropriate patch due to its inability to identify and reason about the interprocedural execution logic of this error. Furthermore, our attempts to generate a patch using only the error report and the error-triggering function \code{clone\_data}, or even the complete calling chain or \texttt{test.c} file, were unsuccessful.

\textbf{\thistool Approach.} In Phase (a), the use-after-free error is replayed and confirmed. Phase (b) identifies the relevant memory operations according to the finite typestate automaton \(\mathcal{A}_\text{UAF}\) in Table~\ref{tab:fsm}, extracting the error-propagation path and relevant program contexts. Phase (c) feeds the error report, context trace, and error-propagation path from the previous phases into the LLM, which infers the correct patch and provides clear explanations for its decisions, as shown at the bottom of Figure~\ref{fig:motivating}. 
The generated patch demonstrates sophisticated interprocedural context and semantics awareness, implementing coordinated modifications across three distinct code segments spanning two different functions. This solution aligns precisely with the ground truth that eliminating this error while preserving logical correctness of the project. 

\begin{figure}[t]
    \centering
        \includegraphics[width=1.0\linewidth]{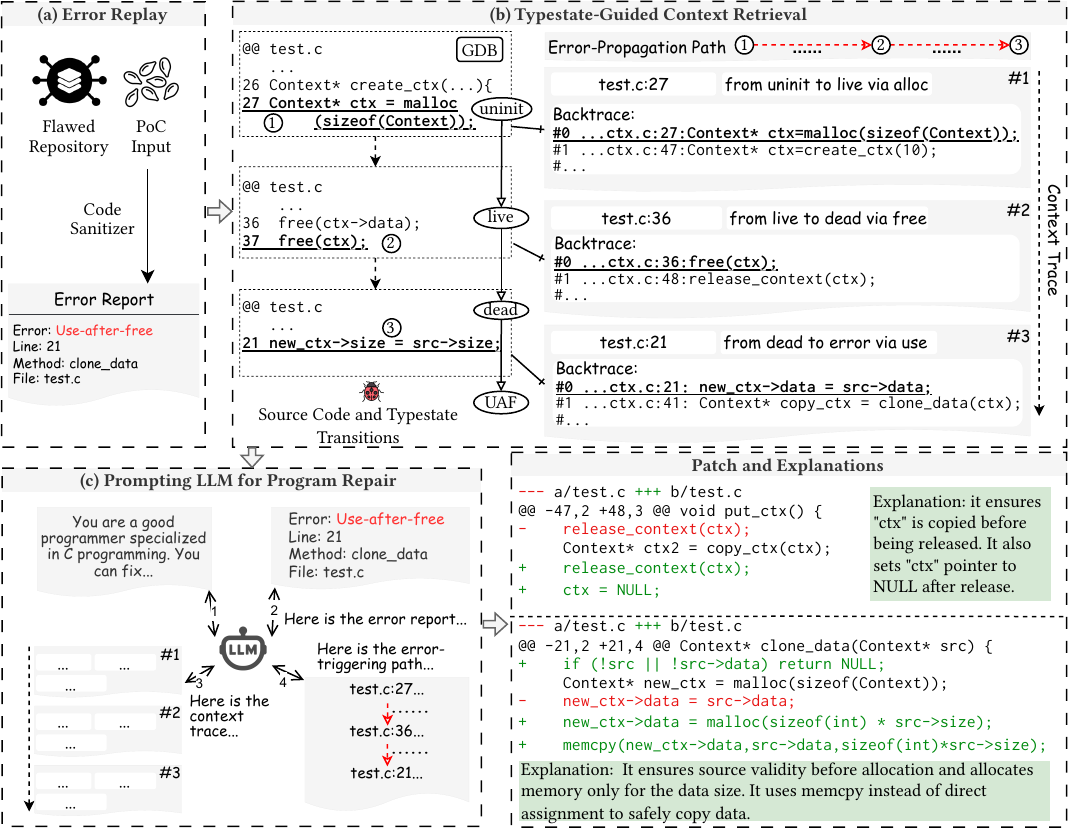}
        \vspace{-5mm}
        \caption{A motivating example illustrating how \thistool repairs a use-after-free error.}
        \vspace{-5mm}
        \label{fig:motivating}
\end{figure}



\textbf{(a) Error Replay.}
Understanding the root cause of a memory error often requires replaying the error. As illustrated in Figure~\ref{fig:motivating}(a), we reproduce this error using its proof of concept (PoC) input, which is a specific assembly file. We employ Valgrind~\cite{Nethercote2007Valgrind}, a widely-used dynamic analysis tool, to capture the core dump of the PoC input at the moment the error is triggered---specifically, at Line 21 in the \texttt{test.c} file. This core dump provides a snapshot of the memory state at the time of the crash. Based on the core dump, we generate an \emph{error report} that offers the specific error type and error location to the LLM.

\textbf{(b) Typestate-Guided Context Retrieval.}
We employ a finite typestate automaton (\(\mathcal{A}_{\text{UAF}}\) in Table~\ref{tab:fsm}) to facilitate the extraction of the \textit{error-propagation path} using the GDB debugger~\cite{GDB}. This path begins with the memory allocation at \textcircled{1}, proceeds through the memory free at \textcircled{2}, and ultimately stops at the error-triggering point at \textcircled{3}, as indicated by an error typestate. 
Along the path, we use \(\mathcal{A}_{\text{UAF}}\) to guide typestate transitions of the erroneously accessed memory object obtained from the previous step. This allows us to monitor the typestate changes and the corresponding contexts at each typestate change breakpoint with the GDB debugger. 
This rigorous tracing enables the LLM to comprehend the structure of the allocated memory object (e.g., \code{Context}), the interprocedural evolution of the error across the code repository, and the broader program semantics (logical correctness) along the execution trace.

The key transitions in this model are triggered by the \code{malloc}, \code{free}, and \code{use} statements, which occur at \textcircled{1}, \textcircled{2}, and \textcircled{3}, respectively. 
Upon invoking the \code{malloc} function, denoted as the \texttt{alloc} operation in our FTA, the typestate of the memory object shifts to \texttt{live}, indicating that the memory is currently in use. Subsequently, when the \code{free} API is called, the typestate changes to \texttt{dead}, signifying that the memory has been released. However, if the same memory object is used after it is released, the typestate changes to \texttt{error}, indicating a use-after-free error.
At each typestate change point, we extract its associated context, which includes the typestate transition, the location, and the backtrace. The program contexts at the three typestate change points collectively form a \emph{context trace}. For instance, in the final error context, the typestate transition indicates a shift from \texttt{dead} to \texttt{error} due to a \texttt{use} operation. The location, \texttt{test.c:21}, specifies the file path and line number. The backtrace, gathered from the debugger at this breakpoint, provides the call stack details, illustrating how the \texttt{use} operation was invoked by \code{clone\_data} and other higher-level callers in the codebase.



\textbf{(c) Prompting LLM for Program Repair.}
In this phase, depicted in Figure~\ref{fig:motivating}(c), we employ structured prompting~\cite{hao2022structured} to break down the prompts into four structured steps: role and task description~\cite{shanahan2023role}, error report, context trace, and error-propagation path. These segments of information are systematically fed to the LLM step by step. This method effectively deconstructs the task of program repair into increasingly detailed and specific stages, allowing the LLM to progressively comprehend and tackle the error.
\section{\thistool Approach\label{sec:approach}}
In this section, we detail our \thistool approach. We first identify the specifics of the memory error (\S\ref{sec:errconfirm}), then utilize typestate-guided context retrieval to understand its semantics (\S\ref{sec:typestatetracing}). Finally, we leverage this information to generate a patch and explanation(s) via LLM (\S\ref{sec:llmapr}).

\subsection{Error Replay\label{sec:errconfirm}}
\label{subsec:errconfirm}

In this paper, our emphasis is not on detecting memory errors but on repairing validated true errors. The first step of our approach involves reproducing the error and identifying the associated bug report, which includes the type of memory error and its location. To achieve this, we utilize a dynamic analysis tool (DAT) such as ASan~\cite{ASan} or Valgrind~\cite{Nethercote2007Valgrind} during compilation and generate debuggable files following the DWARF 5~\cite{DWARF5} standard. We then use specific proofs of concept (PoCs) to reproduce the error and employ DAT to generate detailed error reports. At the error-triggering point, the DAT outputs information about the type of memory error, the location, and the specific memory address that is erroneously accessed (i.e., the error address).

\subsection{Typestate-Guided Context Retrieval\label{sec:typestatetracing}}

This section elaborates on how we enhance LLM prompts with a deep understanding of memory error semantics while optimizing token usage. We achieve this by employing typestate finite state automata (Definition~\ref{def:tfsa}) to guide the derivation of a comprehensive yet precise error-propagation path (Definition~\ref{def:error_path}) and context trace (Definition~\ref{def:ctxtrace}) that capture the interprocedural contextual evolution linked with the memory error. This serves as a key component for patch generation in the subsequent phase.

We first demonstrate how to construct the full execution path of the error (Definition~\ref{def:fullpath}) and the typestate-changing context map (Definition~\ref{def:sinfo}) using Algorithm~\ref{ag:fep-construction}. The context map records the program context (Definition~\ref{def:ctx}) specifically at the statements introducing typestate changes. Using the full path and context map as a foundation, we then illustrate how to build the error-propagation path as explained in the rule in Figure~\ref{fig:etp-rules} and the context trace as outlined in Figure~\ref{fig:cxt-rules}. We also discuss the correctness of each algorithm and inference rule.

\begin{definition}[Full Execution Path $\pi$]
    \label{def:fullpath}
A full execution path, denoted as \(\pi=(s_i)_{i=1}^n\), is a chronologically ordered sequence of program statements, each assigned an index, such that \(s_i=\langle \mathtt{sc}(s_i), i\rangle\). Here, \(\mathtt{sc}(s_i)\) retrieves the source code associated with \(s_i\), while \(i\) represents the position of this statement in the sequence, indicating the execution order. This sequence starts from the program's entry point and extends all the way to the error-triggering statement.
\end{definition}

Note that for any \(s_i, s_j\) in \(\pi\) where \(i\neq j\), their corresponding source code can be identical, i.e., \(\mathtt{sc}(s_i)=\mathtt{sc}(s_j)\), because a single code line can be invoked or executed multiple times throughout the program's run, depending on the control flow of the program.
\LinesNumbered
\SetNoFillComment

\SetKwInput{KwInput}{Input}   
\SetKwInput{KwOutput}{Output}
\SetKwFunction{MAIN}{$\mathtt{constructPi}$}
\begin{center}
    \begin{algorithm}[t]
    \small
    \caption{Full execution path and typestate-changing context map construction.} \label{ag:fep-construction}
    \DontPrintSemicolon
    \KwInput{$\mathit{addr}_e$: Memory address erroneously accessed (error address); $P$: Target program; 
    $I$: PoC input; $\mathcal{A}_\textup{ET}=\left<\Sigma ,\mathbb{T},T_u,\delta,T_\textup{ET} \right>$: Finite typestate automaton;
    }
    \KwOutput{$\pi$: Full execution path; $\tMap$: Map from typestate-changing statements to program contexts
    }
    \SetKwProg{Fn}{Function}{:}{\KwRet}
    \Fn{\MAIN{$\mathit{addr}_e$, P, I, $\mathcal{A}_\textup{ET}$}}{
    
        GDB.execute(``file ''+ $P$);     \agc{Set the program to execute} \\
        GDB.execute(``set arg ''+ $I$);  \agc{Set the program PoC input}\\
        GDB.execute(``start'');        \\
        $\pi \gets ()$; $\tMap \gets \{\}$; $T \gets T_u$; $i\gets 1$; \\
        frame $\gets$ GDB.selected\_frame();\\
        \While{$T\neq T_\textup{ET}$\label{line:endOfE}}{
            $\lab \gets \textup{frame.code\_line}$; \\
            $s \gets \langle \lab, i \rangle$; \\
            \If(\agc{Whether the operation of $s$ belongs to the FTA language and $s$ manipulates the error address}){$\mathtt{op}(s) \in \Sigma\wedge \textup{frame.addr} = \mathit{addr}_e$\label{line:opmatch}}{
                
            \If(\agc{Typestate changes}){$(T' \gets \delta(T, \mathtt{op}(s))) \neq T$\label{line:statechange}}{
                \text{backtrace} $\gets$ GDB.execute("backtrace");\\
                 $\mathit{Ctx}_s \gets \{\textup{frame.location}, \mathtt{transition}(T, T', \mathtt{op}(s)) , \textup{backtrace}\}$;\\
                 $\tMap\gets \tMap\cup \{s\mapsto \mathit{Ctx}_s\}$;\agc{Record typestate-changing context in \tMap}\\
                 $T \gets T'$;\\
                 }
            }
            $\pi \gets \pi \circ s$;\agc{Append $s$ to $\pi$}\\
            $i$++;\\
            \text{GDB.execute("step")};\\
            frame $\gets$ GDB.selected\_frame(); \\
            }
        \KwRet{$\pi, \tMap$};\\
    }
    \end{algorithm}
\end{center}



\begin{definition}[Program Context $\mathit{Ctx}$]
\label{def:ctx}
    \label{def:program_context}
Given a full execution path \(\pi\), for any statement \(s_i \in \pi\), the program context of \(s_i\) is defined as:
\[
\mathit{Ctx}_{s_i} = \langle \mathit{lc}, \mathit{tr}, \mathit{cp} \rangle
\]

where:
\begin{itemize}
\item \(\mathit{lc}\) denotes the program location (file path and line number) of \(s_i\),
\item \(\mathit{tr}\) represents the typestate transition at \(s_i\), encapsulating the pre- and post-typestates and the memory operation, and
\item \(\mathit{cp}\) signifies the backtrace of the call path, showing the call sequence leading up to the current point of execution. Specifically, it records each function call along with its source location and the corresponding source code at that location.
\end{itemize}
\end{definition}

\begin{definition}[Typestate-Changing Context Map $\tMap$]\label{def:sinfo}
A typestate-changing context map $\tMap$ associates statements inducing typestate changes with their respective program contexts. Importantly, each memory object manipulated by the statements in $\tMap$'s key set must be aliased with the object at the error-triggering point. This aliasing is evidenced by their shared error address.
\end{definition}

\runinhead*{Full Execution Path and Typestate-Changing Context Map Construction.} Algorithm \ref{ag:fep-construction} outlines the construction of the full execution path (\(\pi\)) and the typestate-changing context map (\tMap) using the GNU Debugger (GDB)~\cite{GDB}. The algorithm initializes GDB and runs the program with input \(I\). Within the loop (Lines~7-19), it constructs \(\pi\) and \tMap step-by-step until an error state occurs. Each loop iteration, representing a program execution step, appends the statement to \(\pi\) (Line~16), and collects data from the current stack frame. If an operation manipulates the error address and causes a typestate change (Lines~10-11), the algorithm inserts the program context, including the current source code line, typestate transition and backtrace, into \(\tMap\) (Lines~13-14). 

\begin{figure}[t]
\centering
\includegraphics[width=\linewidth]{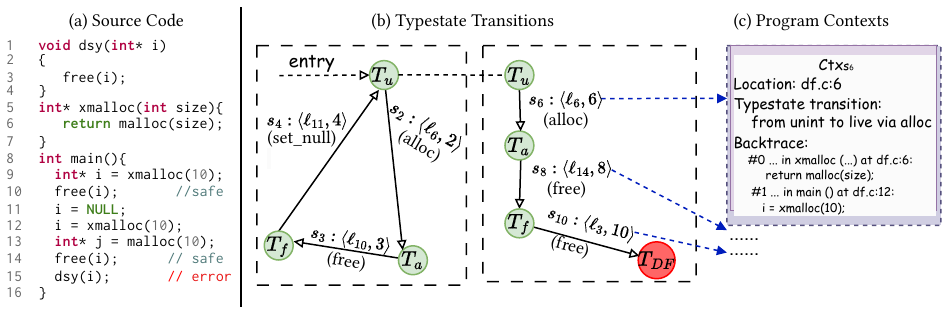}
\vspace{-8mm}
\caption{An example of typestate transitions and program contexts.}
\vspace{-5mm}
\label{fig:approach_example1}
\end{figure}

\begin{example}
\label{ex:1}
Figure~\ref{fig:approach_example1}(a) shows a double-free error at $\lab_3$ (invoked from $\lab_{15}$). 
In this example, the program execution flows through 10 steps as follows:
    \[
        \pi = (s_i)_{i=1}^{10} = \left( \langle \lab_9,1 \rangle, \langle \lab_6,2 \rangle, \langle \lab_{10},3 \rangle, \langle \lab_{11},4 \rangle, \langle \lab_{12},5 \rangle, \langle \lab_6,6 \rangle, \langle \lab_{13},7 \rangle, \langle \lab_{14},8 \rangle, \langle \lab_{15},9 \rangle, \langle \lab_3,10 \rangle \right)
    \]
    where $\lab_n$ refers to the source code at Line $n$ for this example. 
    The process of typestate transition and the corresponding results are illustrated in Figure~\ref{fig:approach_example1}(b). 
    The statements that induce typestate changes are as follows:
    \[
    s_2: \langle \lab_6, 2 \rangle, \quad s_3: \langle \lab_{10}, 3 \rangle, \quad s_4: \langle \lab_{11}, 4 \rangle, \quad s_6: \langle \lab_6, 6 \rangle, \quad s_8: \langle \lab_{14}, 8 \rangle, \quad s_{10}: \langle \lab_{3}, 10 \rangle
    \]
    $s_{7}: \langle \lab_{13}, 7 \rangle$ does not belong to these statements because it operates on a different address from $\mathit{addr}_e$.
    Accordingly, we construct the typestate context map $\tMap$ as follows:
    \[ 
    \tMap = \{s_2 \mapsto \mathit{Ctx}_{s_2}, s_3 \mapsto \mathit{Ctx}_{s_3}, s_4 \mapsto \mathit{Ctx}_{s_4}, s_6 \mapsto \mathit{Ctx}_{s_6}, s_8 \mapsto \mathit{Ctx}_{s_8}, s_{10} \mapsto \mathit{Ctx}_{s_{10}}\}
    \]
    where $\mathit{Ctx}_{s_i}$ denotes the program context at statement $s_i$. 
    The mapping $(s_n \mapsto \mathit{Ctx}_{s_n})$ represents the association between the statement $s_n$ and its corresponding program context $\mathit{Ctx}_{s_n}$, as defined in Definition~\ref{def:sinfo}.
    For instance, the program context $\mathit{Ctx}_{s_6}$, depicted in Figure~\ref{fig:approach_example1}(c), includes the statement's location, detailed typestate transition information, and backtrace data. The backtrace clearly shows the call stack from the \code{main} function to the \code{xmalloc} function as well as the exact line of source code at each stack frame.
    \label{ex:fullpath}
\end{example}

\begin{lemma}[Correctness of Algorithm~\ref{ag:fep-construction}]
\label{lemma:fep}
Algorithm~\ref{ag:fep-construction} correctly constructs $\pi$ according to Definition~\ref{def:fullpath} and \tMap according to Definition~\ref{def:sinfo}.
\end{lemma}

\begin{proof}[Proof Sketch]
In Algorithm~\ref{ag:fep-construction}, Lines 2-4 ensure that \( s_1 \) is the program's entry point. The condition in Line 7 guarantees loop termination at the error state, which indicates an error-triggering statement (according to Definition~\ref{def:tfsa}). Lines 18-19 ensure that the loop follows the program execution order. Thus, the construction of \(\pi\) aligns with Definition~\ref{def:fullpath}. Additionally, the conditions in Lines 10 and 11 conform to Definition~\ref{def:sinfo} for constructing \tMap.
\end{proof}

\begin{definition}[Error-Propagation Path $\pi_e$]
    \label{def:error_path}
An error-propagation path \(\pi_e = (s_i)_{i=m}^n\) is a subsequence of indexed program statements that extends from the memory allocation statement $s_m$, where $0\leq m<n \wedge \mathtt{op}(s_m)=\text{alloc}$, to the error-triggering point $s_n$. 
$s_m$ is the closest memory allocation to $s_n$ on the full path $\pi$ where the allocated memory address is the error address accessed at $s_n$.
\end{definition}

\runinhead*{Typestate-Guided Error-Propagation Path Extraction.}
The inference rule depicted in Figure~\ref{fig:etp-rules} is used to extract the error-propagation path, \(\pi_e\), from the full execution path, \(\pi\) (Definition \ref{def:fullpath}). This rule helps reduce the code lines and isolate the error code from the executed code. It works by identifying a particular statement, \(s_m\), associated with an operation that changes the typestate allocation, guided by the \(\tMap\), which tracks these typestate changes (as per Definition \ref{def:program_context}). The rule ensures that any statements on the execution path \(\pi\) falling between \(s_m\) and \(s_n\) do not manipulate the error address or involve memory allocation. This isolates the statements likely to be the root cause of the error, similar to a debugging process that narrows down and focuses only on problematic code segments, effectively reducing the length of the message input into the large language model.

\begin{figure}[t]
\small
\centering
\ruledef{ETP}{
  \pi=(s_i)_{i=1}^n\  s_m \in \mathtt{KS}(\tMap) \wedge \mathtt{op}(s_m)=\text{alloc} \wedge (\forall k, m < k \leq n: s_k \notin \mathtt{KS}(\tMap) \vee \mathtt{op}(s_k) \neq \text{alloc})
}{
 \pi_e \gets(s_i)^n_{i=m}
}
\vspace{-3mm}
\caption{The inference rule for typestate-guided error-propagation path extraction. $\mathtt{KS}(\tMap)=\{s_i\mid (s_i\mapsto \mathit{Ctx_{s_i}}) \in\tMap\}$ represents the key set of $\tMap$.}
\vspace{-3mm}
\label{fig:etp-rules}
\end{figure}

\begin{example}
The error-propagation path of Example~\ref{ex:fullpath} is:
\[\pi_e = (s_i)_{i=6}^{10} = \left( \langle \lab_6,6 \rangle, \langle \lab_{13},7 \rangle, \langle \lab_{14},8 \rangle, \langle \lab_{15},9 \rangle, \langle \lab_3,10 \rangle \right)\]
The path starts from $s_6$, which is the nearest allocation operation preceding the error-triggering point $s_{10}$ on $\pi$. Although \(s_7:\langle\lab_{13},7\rangle\) lies between these two points, it is not designated as the starting point because the variable \(j\) does not point to the error address. The operation at \(s_8:\langle\lab_{14},8\rangle\) does manipulate the error address, but it does not qualify as a starting point because it does not involve a memory allocation operation.
    \label{ex:ctxtrace}
\end{example}


\begin{proof}[Correctness of Rule \texttt{[ETP]}]
According to Definition~\ref{def:sinfo} and Lemma~\ref{lemma:fep}, \(\tMap\) records only the statements related to the erroneously accessed memory address $\mathit{addr}_e$. 
Therefore, the constraints in the rule guarantee that $s_m$ is the error memory allocation ($s_m \in \mathtt{KS}(\tMap) \wedge \mathtt{op}(s_m)=\text{alloc}$) that is the closest to $s_n$ on $\pi$ $(\forall k, m < k \leq n: s_k \notin \mathtt{KS}(\tMap) \vee \mathtt{op}(s_k) \neq \text{alloc})$.
\end{proof}




\begin{figure}[t]
    \centering
    \includegraphics[width=\linewidth]{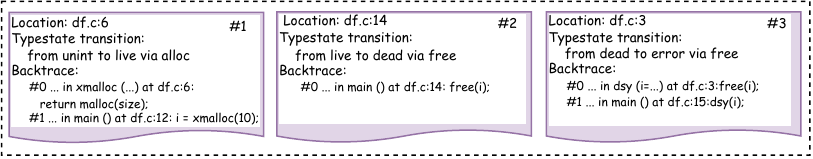}
    \vspace{-8mm}
    \caption{An example of context trace by revisiting Example~\ref{ex:1}.}
    \vspace{-5mm}
    \label{fig:approach_df-trace}
    \end{figure}

\begin{definition}[Context Trace $\widetilde{\mathit{Ctx}}_e$]
\label{def:ctxtrace}
The context trace of an error-propagation path $\pi_e$ is formally defined as $\widetilde{\mathit{Ctx}}_e = ( \mathit{Ctx}_{s_i} \mid s_i \in \pi_e \land \mathtt{Sel}(s_i))$, where each $\mathit{Ctx}_{s_i}$ is the program context at $s_i$, and the sequence follows the order of statements in $\pi_e$. 
Without loss of generality, the selection function $\mathtt{Sel}$ is a predicate over the statements in $\pi_e$.
In our definition, $\mathtt{Sel}(s_i)$ returns true if $s_i$ introduces typestate changes of the memory object operated at the error-triggering point $e$.
\end{definition}

\begin{figure}[ht]
\small
\centering
\ruledef{CXT}{
  s_i\in\pi_e\quad s_i\in\mathtt{KS}(\mathit{tMap})
}{
 \widetilde{\mathit{Ctx}}_e\gets \widetilde{\mathit{Ctx}}_e\circ\tMap[s_i]
}
\vspace{-3mm}
\caption{The inference rule for typestate-guided context trace construction. $\tMap[s_i]$ retrieves the program context related to $s_i$ in $\tMap$.}
\vspace{-3mm}
\label{fig:cxt-rules}
\end{figure}

\runinhead*{Typestate-Guided Context Trace Construction.}
Figure~\ref{fig:cxt-rules} shows the inference rule for constructing the context trace, denoted as $\widetilde{\mathit{Ctx}}_e$, which captures the program contexts associated with typestate-changing statements along an error-propagation path $\pi_e$. 
The rule iterates through each statement $s_i$ in $\pi_e$ from index $m$ to $n$. For each statement $s_i$, if it is a typestate-changing operation, then the corresponding context is appended to $\widetilde{\mathit{Ctx}}_e$. 


\begin{example}
    Figure~\ref{fig:approach_df-trace} presents the context trace of Example~\ref{ex:ctxtrace}. 
    By iterating through the error-propagation path, we append the relevant context information to the trace whenever a typestate-changing operation is encountered in $\tMap$.
    The resulting context trace, depicted in Figure~\ref{fig:approach_df-trace}, consists of three program contexts, clearly showing the lifecycle of the memory, starting from its allocation, through its release, and finally to its erroneous access.
\end{example}

\begin{proof}[Correctness of Rule \texttt{[CXT]}]
The rule ensures the order of context trace follows $\pi_e$ ($s_i\in\pi_e$), and also matches the semantics of \texttt{Sel} in Definition~\ref{def:ctxtrace} ($s_i\in\mathtt{KS}(\tMap)$) according to the definition of $\tMap$ in Definition~\ref{def:sinfo} and Lemma~\ref{lemma:fep}.
\end{proof}

\subsection{Prompting LLM for Error Repair\label{sec:llmapr}}
This section discusses how we utilize the information collected from the typestate-guided analysis to prompt LLM for efficient memory error repair. The process begins by employing the role-play technique~\cite{shanahan2023role} to define and motivate our task, followed by the use of structured prompting~\cite{hao2022structured} to direct the LLM in crafting a patch and explanation(s).

\runinhead*{Role Play.}
Role play is an effective strategy known to enhance the precision of LLM outputs~\cite{shanahan2023role}. In our case, we configure the LLM to function as an APR tool, where it interprets error information from the typestate-guided analysis and generates an appropriate patch with intuitive explanations. 

\runinhead*{Structured Prompting.}
We structure prompts to guide the LLM in generating patches by providing the necessary information to know the error location, understand the error context and path, and then create a resolution. 
The error reports, along with the context trace and error-propagation path guided by the typestate analysis, are formatted as prompts and fed into the LLM. 

The error report guides the LLM in pinpointing the exact location of the memory error in the codebase (\S\ref{sec:errconfirm}). The initiation of the error report typically begins with a statement like, ``Here is the location of the use-after-free error in the provided code snippet''.
The second component is the context trace $\widetilde{\mathit{Ctx}}_e$ (Definition~\ref{def:ctxtrace}) that captures the context sequence of memory state transitions leading to the error, enabling the LLM to comprehend the execution logic of the error.
The final component presents the error-propagation path $\pi_e$ (Definition~\ref{def:error_path}) to the LLM. This path, a critical subsequence of the full execution path $\pi$, extends from the nearest memory allocation statement to the error-triggering point on $\pi$, showing the program dependencies along an intact erroneous memory management. 
Upon presenting all the necessary information for understanding the error semantics and triggering logic, the LLM is then tasked with generating a patch to repair the error.

\section{Evaluation\label{sec:evaluation}}
This section evaluates \thistool's performance in repairing memory errors in real-world projects by comparing it with two state-of-the-art memory error APR tools: \SAVER~\cite{SAVER} and \ProveNFix~\cite{ProveNFix}. \thistool successfully repairs \compares and \comparep more memory errors than \SAVER and \ProveNFix, respectively, while introducing no new errors. We also compare \thistool with LLM-based baselines and conduct an ablation analysis to understand the contribution of each component.

\subsection{Datasets and Implementation}

\begin{table}[t]
\small
    \caption{The statistics and description of the real-world projects used in the evaluation. LoC stands for lines of code. \#File and \#Error represent the number of files and memory errors, respectively.}
    \vspace{-3mm}
    \begin{adjustbox}{width=1\textwidth, center}
    \begin{tabular}{|l|l|r|r|r|l|} \hline
        \textbf{No.} & \textbf{Project (version)}& \textbf{LoC}& \textbf{\#File} &\textbf{\#Error} &\textbf{Description} \\  \hline
        1 & \texttt{ls\_extended} (9d899c8)~\cite{ls}    &   1,352     &   25     &  3 & ls with coloring and icons                    \\  \hline
        2 &\texttt{xHTTP} (72f82d)~\cite{xHTTP}          & 1,493      &   6                & 1   & HTTP server library           \\  \hline
        3 &\texttt{tree} (v1.8)~\cite{tree}              &   3,435   &   13    & 3 & utility to display a tree view of folders                     \\  \hline
        4 &\texttt{chibicc} (90df7f)~\cite{chibicc}      &   9,688   &   71      &  5  & C compiler                   \\  \hline
        5 &\texttt{stb} (v2.8)~\cite{stb}                &   12,076  &   2       &   3 & single-file public domain libraries for C/C++                    \\ \hline
        6 &\texttt{scrot} (b5e5f0d)~\cite{scrot}         &   13,130     &   36     & 1 & command line screen capture utility                    \\  \hline
        7 &\texttt{mjs} (b1b6eac)~\cite{mjs}             &   32,116  &   202    &   1   & embedded JavaScript engine                   \\ \hline
        8 &\texttt{SmallerC} (b120a9c)~\cite{SmallerC}   &   58,535  &   510    &   14 & C compiler                     \\ \hline
        9 &\texttt{MyHTML} (90a853e)~\cite{myhtml}        &   63,617    &   168      & 2  & Fast C/C++ HTML 5 Parser                   \\  \hline
        
        10 &\texttt{quickjs} (d378a9f)~\cite{quickjs}        &   86,281  &   53     &   2 & embedded JavaScript engine                      \\  \hline
 
        11 &\texttt{recutiles} (v1.8)~\cite{recutiles}    &   92,000    &   757    &   5 & tools and libraries to access recfiles                     \\  \hline    
        12 &\texttt{wasm3} (139076a)~\cite{wasm3}        &   111,616    &   696      & 4  & WebAssembly interpreter                   \\  \hline
        13 &\texttt{Yasm} (ffbd22c)~\cite{Yasm}          &   201,975   &   945    &   4  & assembler                    \\  \hline
        14 &\texttt{radare2} (8644a29)~\cite{radare2}    &   879,785  &   2,953   &   1 & reverse engineering framework                     \\  \hline
        \multicolumn{2}{|l|}{\textit{\textbf{Total}}}  & 1,567,099 &   6,437  &   \totalmemerr &    \\  \hline
    \end{tabular}
    \end{adjustbox}
    \label{tab:projects}
    \vspace{-5mm}
    \end{table}

\runinhead*{Datasets.} 

We first compare \thistool with \SAVER and \ProveNFix using the same dataset as \SAVER~\cite{SAVER}. We meticulously reverse-engineer the vulnerability triggering conditions from \SAVER's vulnerability set and construct proof of concept inputs to reproduce identical errors. From the original collection, we exclude multi-threaded programs because our debugger-based approach cannot effectively trace typestate transitions across multiple execution contexts, resulting in eight projects for comparative evaluation.
However, we identify limitations in \SAVER's dataset, including its homogeneity and insufficient representation of real-world vulnerabilities. For instance, the same version of recutiles~\cite{recutiles} contains CVE-2019-6455~\cite{CVE-2019-6455}, which is absent from \SAVER's dataset. Furthermore, fixes in \SAVER's dataset rely solely on manual verification without the crucial validation provided by developer acceptance in real-world scenarios.

To address these limitations, we further construct a comprehensive benchmark comprising \pronum real-world, open-source C projects, encompassing approximately 1.57 million lines of code across diverse application domains including system utilities, network protocols, and media processing libraries. Table~\ref{tab:projects} presents detailed statistics and descriptions for these projects. All memory errors in our benchmark have been \emph{confirmed by respective project developers}, with 9 vulnerabilities assigned CVE identifiers and others documented in official commit histories. To establish reliable ground truth for validating correct fixes, we collect patches approved and implemented by project maintainers. Additionally, we curate error-inducing test cases (proof of concept inputs) to consistently reproduce each error during evaluation.

\runinhead*{Implementation.}
Our experiments are conducted on an Ubuntu 22.04 server with a 24-cores 5.60 GHz Intel CPU and 64 GB of memory. We employ Valgrind (version 3.18.1)~\cite{Nethercote2007Valgrind} and ASan~\cite{ASan} for error replay (\S\ref{sec:errconfirm}). We use the GDB Python API (with GDB version 12.1) and the GNU C Library (version 2.35) to perform typestate-guided context retrieval (\S\ref{sec:typestatetracing}). 
To reduce the time overhead for extracting error-propagation paths and context traces, we set the initial breakpoint at memory allocations instead of the program's entry point. We use Claude 3.5 Sonnet~\cite{Claude, Memory_Safety_Blog} as our LLM (\S\ref{sec:llmapr}). For each prompt, we apply the chain-of-thought~\cite{CoT} method to encourage the model to think step-by-step, thereby generating coherent and logically consistent responses.
Given the inherent randomness of LLMs, it is necessary to ensure the reliability of our results. Therefore, we conduct each experiment five times for every project. We only consider a result to be valid if it exhibits consistency in at least four out of the five runs.


\subsection{Baselines\label{sec:baselines}}

We compare \thistool with six baselines spanning traditional state-of-the-art memory APR tools and LLM-based approaches:

\runinhead*{\SAVER and \ProveNFix.}
To the best of our knowledge, \SAVER~\cite{SAVER} and \ProveNFix~\cite{ProveNFix} represent the current state-of-the-art in memory error APR, capable of addressing memory leaks, use-after-frees, and double-frees. We employ their open-source implementations for comparative analysis. These tools rely on detection results from the static analysis tool \Infer~\cite{infer}. To ensure fair comparison, we augment these tools with \emph{precise error locations}, identical to those provided to \thistool, to direct targeted repairs. Futhermore, to mitigate potential limitations in static analysis, we configure \SAVER and \ProveNFix with \emph{the most advanced parameters} for \Infer as described in their respective publications. This configuration encompasses whole-program analysis of the linked program, flow-sensitive analysis distinguishing control flows, and comprehensive header file parsing~\cite{SAVER, ProveNFix}. 

\runinhead*{SWE-agent 1.0.} 
Existing LLM-based APR tools primarily target Java and Python bugs and predominantly address single-hunk program repairs~\cite{xin2024multihunk}. 
We evaluate \thistool against SWE-agent 1.0~\cite{yang2024sweagent} (hereafter referred to as SWE-agent), an open-source  state-of-the-art LLM-based approach that employs Chain-of-Thought reasoning and sophisticated prompt engineering to autonomously invoke tools for error repair.
To ensure methodological consistency and experimental validity, we provide SWE-agent with identical reproducible errors, comprising comprehensive reproduction workflows, toolchains, proof-of-concept inputs, compiled error versions, and precise trigger and compilation commands. 

\runinhead*{\thistool-F and \thistool-M.} 
To evaluate the capabilities of the base LLM used in \thistool without any contextual retrieval from error trace, we implement two comparative baselines. The first baseline, \thistool-F, incorporates the error report alongside program files directly implicated in the error-triggering point. The second baseline, \thistool-M, extends this approach by including all functions identified in the error backtrace in addition to the error report. Both baselines employ the same structured role-play and prompting techniques delineated in \S\ref{sec:llmapr} that are utilized in \thistool, thus isolating the effect of context retrieval.

\runinhead*{\thistool-NT.} 
To understand the relative contributions of typestate-guided context retrieval to the overall performance of \thistool, we conduct comprehensive ablation studies using \thistool-NT (\thistool without context trace). \thistool-NT omits the context trace in the prompts while maintaining all other components unchanged (such as bug report and error-propagation path), allowing us to measure the specific impact of typestate-guided context trace information.

    
    
    

\subsection{Evaluation Metrics}
We collect the number of patches generated by each APR tool, denoted as \(\#\Delta\). The number of correct patches is represented as \(\#\Delta_{\checkmark}\). To be classified as correct, a patch must: 1) successfully fix the memory error; 2) preserve the expected outcomes for the project's test suite; and 3) be manually validated to align with the ground truth, adhering to the standards outlined in~\cite{liu2020eval}. We use \(\#\Delta_{\text{\ding{55}}}\) to denote the count of patches introducing new errors, as confirmed through fuzzing tests~\cite{Fioraldi2020afl++}. Conversely, \(\#\Delta_{\text{O}}\) denotes the quantity of patches that, while failing to correct the error, do not introduce new ones. Finally, \(\#E_{\checkmark}\) represents the number of fixed errors. 
In terms of performance evaluation, a higher value for both \(\#\Delta_{\checkmark}\) and \(\#E_{\checkmark}\) suggests superior performance, as it implies a greater number of errors have been correctly repaired. Conversely, \(\#\Delta_{\text{\ding{55}}}\) should ideally be minimized to avoid the introduction of new errors. Although \(\#\Delta_{\text{O}}\) does not pose a severe risk of introducing new memory errors, a smaller value is still preferable.
Note that we count memory errors by logical "blocks" rather than by individual pointers. For memory leaks, each distinct allocation of a data structure that fails to be deallocated is counted as a separate error, as each requires an independent deallocation operation. This block-based accounting methodology accurately reflects the granularity at which memory management operations must be applied in practice. 
Consequently, a single patch that properly deallocates a complex data structure may fix multiple memory error blocks simultaneously. Therefore, it is possible that the number of correct patches \(\#\Delta_{\checkmark}\) is smaller than the number of fixed errors \(\#E_{\checkmark}\).


\subsection{Research Questions}
We address the following research questions in our evaluation:
\begin{enumerate}
    \item[RQ1] \textbf{Comparison with traditional APR tools:} How effectively does \thistool repair memory errors compared to state-of-the-art tools such as \SAVER and \ProveNFix, in terms of both repair accuracy and error introduction?
    \item[RQ2] \textbf{Comparison with LLM-based approach:} To what extent does \thistool improve repair performance compared to an open-source state-of-the-art LLM-based approach SWE-agent~\cite{yang2024sweagent}?
    \item[RQ3] \textbf{Ablation analysis:} What is the relative contribution of typestate-guided context retrieval to \thistool's overall repair effectiveness?
\end{enumerate}

\subsection{Comparison with Traditional APR Tools (RQ1)}

\subsubsection{Comparison Results.} 
Table~\ref{tab:saver_dataset} presents the detailed results of our comparative evaluation with state-of-the-art memory error APR tools \SAVER and \ProveNFix using \SAVER's dataset. Our analysis demonstrates that \thistool significantly outperforms both baseline tools, successfully addressing 122 out of 153 errors. The superior repair capabilities of \thistool are consistent across all projects in the comparative evaluation.
Table~\ref{tab:rq1} provides a comprehensive comparison of \thistool against \SAVER and \ProveNFix on our benchmark of \pronum real-world projects. The empirical evidence clearly establishes \thistool's superior performance in error repair. Our approach successfully repairs \fixmemerr out of \totalmemerr identified memory errors, representing \compares more repairs than \SAVER and \comparep more than \ProveNFix. Notably, \thistool achieves these substantial improvements without introducing any new errors, unlike the baseline tools. The project-level analysis further confirms \thistool's consistency, as it repairs at least as many errors as the baseline tools across all projects while maintaining nearly zero error introduction rates.

\begin{table}[t]
    \centering
    \small
    \caption{Comparing \thistool with \SAVER and \ProveNFix using \SAVER's dataset~\cite{SAVER}.}
    \vspace{-3mm}
    \begin{tabular}{l|c|c|c|cc}
    \toprule
    \textbf{Project }& \textit{\#E} & \SAVER~\cite{SAVER} & \textit{\#E} & \ProveNFix~\cite{ProveNFix} & \textbf{\thistool} \\
    \midrule
        \texttt{rappel} (ad8efd7) & 1 & 1 & 1 & 1 & 1 \\ 
        \texttt{lxc} (72cc48f) & 3 & 3 & 23 & 22 & 22 \\ 
        \texttt{WavPack} (22977b2) & 1 & 0 & 12 & 12 & 12 \\ 
        \texttt{flex} (d3de49f) & 3 & 0 & 4 & 4 & 4 \\ 
        \texttt{p11-kit} (ead7ara)  & 33 & 24 & 28 & 27 & 27 \\ 
        \texttt{recutils} (v1.8) & 10 & 8 & 42 & 36 & 37 \\ 
        \texttt{snort} (v2.9.13) & 16 & 10 & 42 & 13 & 18 \\ 
        \texttt{grub} & 0 & 0 & 1 & 1 & 1 \\ 
        \hline
    \textbf{\textit{Total (Fixing Rate)}} & 67 & 46 (68\%) & 153 & 116 (75.8\%) & \textbf{122 (79.7\%)} \\
    \bottomrule
    \end{tabular}
    \label{tab:saver_dataset}
\end{table}

\begin{table}[t]
    \small
    \centering
    \caption{Comparing \thistool with \SAVER and \ProveNFix using the \pronum real-world projects in Table~\ref{tab:projects}. 
    For simplicity, we omit the project versions.
    $\#E$ denotes the number of memory errors.
    }
    \vspace{-3mm}
    \begin{adjustbox}{width=1\textwidth, center}
\def\arraystretch{1.0}
\setlength{\tabcolsep}{1.1ex}
    \begin{tabular}{|l|r|r|r|r|r|r|r|r|r|r|r|r|r|r|r|r|}
    \hline
    \multicolumn{2}{|c|}{} & \multicolumn{5}{c|}
    {\SAVER~\cite{SAVER}} & \multicolumn{5}{c|}{ProveNFix~\cite{ProveNFix}} & \multicolumn{5}{c|}{\textbf{\thistool}} \\ \hline
        \textbf{Project} & $\#E$ 
    & $\#\Delta$ & $\#\Delta_{\checkmark}$ & $\#\Delta_{\text{O}}$ & $\#\Delta_{\text{\ding{55}}}$ & $\#E_{\checkmark}$ 
    & $\#\Delta$ & $\#\Delta_{\checkmark}$ & $\#\Delta_{\text{O}}$ & $\#\Delta_{\text{\ding{55}}}$ & $\#E_{\checkmark}$ 
    & $\#\Delta$ & $\#\Delta_{\checkmark}$ & $\#\Delta_{\text{O}}$ & $\#\Delta_{\text{\ding{55}}}$ & $\#E_{\checkmark}$ \\ \hline
        \texttt{ls\_extended} & 3 & 1 & 1 & 0 & 0 & 1 & 1 & 1 & 0 & 0 & 3 & 1 & 1 & 0 & 0 & \textbf{3}  \\ \hline
        \texttt{xHTTP} & 1 & 0 & 0 & 0 & 0 & 0 & 1 & 0 & 1 & 0 & 0 & 1 & 1 & 0 & 0 & \textbf{1}  \\ \hline
        \texttt{tree} & 2 & 0 & 0 & 0 & 0 & 0 & 0 & 0 & 0 & 0 & 0 & 1 & 1 & 0 & 0 & \textbf{2}  \\ \hline
        \texttt{chibicc} & 5 & 0 & 0 & 0 & 0 & 0 & 2 & 1 & 0 & 1 & 3 & 3 & 2 & 1 & 0 & \textbf{5}  \\ \hline
        \texttt{stb} & 3 & 0 & 0 & 0 & 0 & 0 & 0 & 0 & 0 & 0 & 0 & 3 & 3 & 0 & 0 & \textbf{3}  \\ \hline
        \texttt{scrot} & 1 & 0 & 0 & 0 & 0 & 0 & 1 & 0 & 1 & 0 & 0 & 1 & 1 & 0 & 0 & \textbf{1}  \\ \hline
        \texttt{mjs} & 1 & 1 & 1 & 0 & 0 & 1 & 1 & 1 & 0 & 0 & 1 & 1 & 1 & 0 & 0 & \textbf{1}  \\ \hline
        \texttt{smallC} & 14  & 1 & 0 & 0 & 1 & 0 & 2 & 1 & 0 & 1 & 1 & 7 & 4 & 2 & 1 & \textbf{12}  \\ \hline
        \texttt{MyHTML} & 2 & 0 & 0 & 0 & 0 & 0 & 1 & 1 & 0 & 0 & 2 & 1 & 1 & 0 & 0 & \textbf{2}  \\ \hline
        \texttt{quickjs} & 2 & 0 & 0 & 0 & 0 & 0 & 0 & 0 & 0 & 0 & 0 & 2 & 2 & 0 & 0 & \textbf{2} \\ \hline
        \texttt{recutiles} & 6 & 0 & 0 & 0 & 0 & 0 & 0 & 0 & 0 & 0 & 0 & 2 & 2 & 0 & 0 & \textbf{4}  \\ \hline
        \texttt{wasm3} & 4 & 0 & 0 & 0 & 0 & 0 & 0 & 0 & 0 & 0 & 0 & 1 & 1 & 0 & 0 & \textbf{2}  \\ \hline
        \texttt{yasm} & 4 & 0 & 0 & 0 & 0 & 0 & 0 & 0 & 0 & 0 & 0 & 2 & 1 & 1 & 0 & \textbf{2}  \\ \hline
        \texttt{radare2} & 1 & 0 & 0 & 0 & 0 & 0 & 1 & 1 & 0 & 0 & 1 & 1 & 1 & 0 & 0 & \textbf{1}  \\ \hline
        \textbf{\textit{Total}} & \totalmemerr & 3 & 2 & 0 & 1 & 2 & 10 & 6 & 2 & 2 & 11 & 27 & 22 & 4 & 1 & \textbf{\fixmemerr}  \\ \hline
    \end{tabular}
    \label{tab:rq1}
    \end{adjustbox}
    \vspace{-3mm}
\end{table}

\subsubsection{Case Study.} 
\label{case-study}
We showcase the superior performance of our tool compared to \SAVER and \ProveNFix through four typical code scenarios depicted in Figure~\ref{fig:case-study}.

\runinhead*{Complex Data Structure.} Figure~\ref{fig:case-study}(a) demonstrates \thistool's in-context repair capability to handle complex data structure. This example presents a memory leak issue where primitive deallocation functions (\code{free(db)}) are inappropriately used for complex data structures. \SAVER and \ProveNFix, despite their sophisticated constraint-based and specification-driven methodologies, fail to generate patches because their analysis erroneously concludes that memory deallocation is already addressed—the code already contains deallocation logic (\code{free(db)}), which superficially satisfies basic memory management constraints.
These tools typically excel at adding missing deallocation functions when none exist but struggle with identifying and replacing inadequate deallocation implementations that fail to account for \emph{complex data structures with nested memory allocations}.

In contrast, \thistool successfully identifies that the existing deallocation strategy is insufficient and generates the correct fix by replacing \code{free(db)} with the structure-specific \code{rec\_db\_destroy(db)} function. This solution properly addresses the memory leaks by recursively deallocating all internal components of the \code{db} structure before releasing the main object itself.
The distinctive performance of \thistool stems from its advanced \emph{context-awareness mechanisms} that not only detect missing deallocations but also evaluate the adequacy of existing memory management code against codebase-specific patterns. By analyzing \emph{allocation-deallocation pairings} throughout the codebase and understanding the structural complexity of data types, \thistool can determine when simple deallocation functions are insufficient and recommend appropriate domain-specific alternatives. Furthermore, since the allocation and deallocation of the object \emph{occur in different functions}, \thistool's interprocedural memory error semantics understanding allows it to track the object's lifecycle across function boundaries, enabling the generation of semantically integrated patches that correctly address the memory leak.

\begin{figure}[t]
    \centering
    \includegraphics[width=0.95\linewidth]{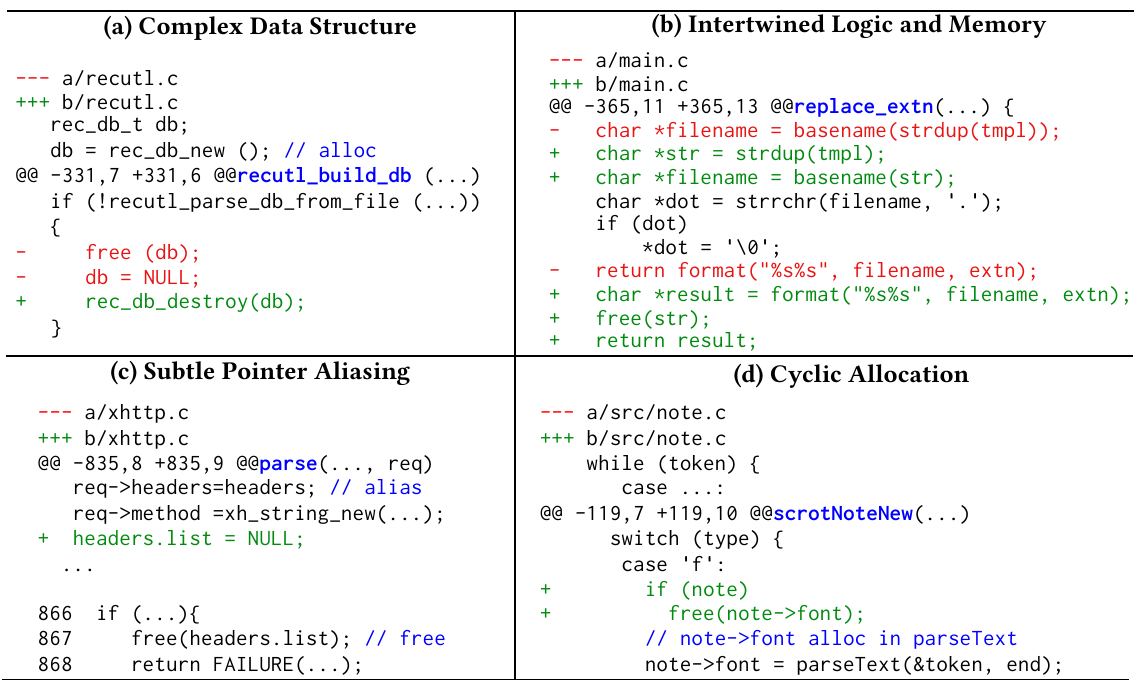}
    \vspace{-5mm}
    \caption{\thistool's patches for: (a) a memory leak in the \texttt{recutiles}~\cite{recutiles} project; (b) a memory leak in the \texttt{chibicc}~\cite{chibicc} project; (c) a double-free vulnerability (CVE-2023-38434~\cite{CVE-2023-38434}) in the \texttt{xHTTP}~\cite{xHTTP} project; and (d) a memory leak error in the \texttt{scrot}~\cite{scrot} project.}  
    \vspace{-5mm}
    \label{fig:case-study}
\end{figure}

\runinhead*{Intertwined Logic and Memory.}
Figure~\ref{fig:case-study}(b) illustrates an example of intertwined logic and memory repair that presents significant challenges for \SAVER and \ProveNFix. Addressing this memory leak necessitates a three-step coordinated modification across disjoint code regions: (1) \emph{introducing a temporary variable} (\code{str}) to track memory allocated by \code{strdup(tmpl)}, (2) \emph{inserting deallocation instructions} (\code{free(str)}) before the function returns, and (3) \emph{restructuring the control/data flow} by placing \code{free(str)} after the use of \code{filename} and storing the return value in an intermediate variable \code{result} rather than immediately returning it. 
This case represents a complex interdependence between program logic and memory management, where an incorrect modification to either aspect could result in functional errors.

In this scenario, \thistool demonstrates superior performance by accurately identifying and resolving both the file handling logic and memory management issues. 
Neither \SAVER nor \ProveNFix could repair the error within the original code structure where memory allocation occurs in a nested function call. To further evaluate their capabilities, we manually restructured the code by splitting \code{char *filename = basename(strdup(tmpl));} at origin $\ell_{365}$ into two distinct statements: \code{char *str = strdup(tmpl);} and \code{char *filename = basename(str);}. With this simplification, both \SAVER and \ProveNFix could generate patches trying to address the memory leak. However, they erroneously positioned the \code{free(str);} statement before the \code{return format("\%s\%s", filename, extn);} at origin $\ell_{369}$, introducing a use-after-free vulnerability since \code{filename} references memory within \code{str}, which would be prematurely deallocated. \thistool, conversely, correctly determines that memory deallocation must occur after all uses of dependent pointers derived from the allocated memory, and generates a semantically sound patch that preserves functionality while eliminating the memory leak.

\runinhead*{Subtle Pointer Aliasing.}
Figure~\ref{fig:case-study}(c) demonstrates \thistool's capability in addressing subtle pointer aliasing scenarios. This vulnerability exemplifies a subtle memory management issue where two pointers, \code{headers} and \code{req->headers} at $\ell_{835}$, reference the same memory region, potentially resulting in double deallocation when memory is freed in both the \code{close\_connection()} and \code{parse()} functions. The challenge here lies not in the syntactic complexity of the code, but in the sophisticated semantic understanding required to trace pointer relationships across function boundaries and execution paths. Such scenarios necessitate comprehension of pointer aliasing to properly identify shared memory references across different control points and contexts.

\SAVER reports "failed to convert labeling operators," indicating its inability to identify the appropriate location for memory deallocation. \ProveNFix generates a patch that introduces an early return before the first \code{free()} operation at $\ell_{867}$, rather than implementing the necessary memory lifecycle management by setting essential pointer to NULL after deallocation.
This approach, while preventing the immediate crash, introduces premature termination of function execution, potentially leading to resource leaks and incomplete functionality. 
\ProveNFix excels at identifying error conditions but struggles with synthesizing correct fixes that require understanding subtle memory sharing semantics. Consequently, \ProveNFix opts for a conservative solution that avoids error manifestation rather than addressing the underlying memory management issue at the exact program point where nullification is required.
In contrast, \thistool comprehends the aliasing relationships between the pointers through the context trace, enabling it to generate a precise and contextually appropriate patch. It correctly inserts a nullification for \code{headers.list} after the pointer copy operation, rather than the conventional practice of inserting NULL after freeing, thereby effectively resolving the double-free vulnerability. 

\runinhead*{Cyclic Allocation.}
Figure~\ref{fig:case-study}(d) illustrates \thistool's effectiveness in addressing cyclic allocation scenarios. 
The error manifests within a while loop where, upon consecutive iterations through the same conditional branch, the program repeatedly allocates memory for the \code{note->font} pointer without first deallocating previously allocated memory resources. This implementation deficiency creates \emph{orphaned memory blocks}, as each subsequent allocation overwrites the reference to previously allocated memory without proper deallocation, thereby causing a persistent memory leak.
Unlike conventional memory management errors where resources simply remain unreleased, this particular error requires understanding that deallocation must precede reallocation within a cyclic execution context.
The non-local nature of the required fix—inserting logic before allocation rather than at the typical post-operation release points—exceeds the pattern-matching capabilities of \SAVER and \ProveNFix.
The patch generated by \thistool establishes a robust resource management guard pattern that first validates the existing resource allocation status \code{if (note)} (preventing potential double-free vulnerabilities) followed by explicit deallocation \code{free(note->font)} prior to subsequent memory allocation operations within the same scope.


\begin{answerBox}
    \textbf{Answer to RQ1}: \thistool outperforms traditional state-of-the-art memory error APR tools \SAVER and \ProveNFix across all projects in both \SAVER's dataset and our comprehensive benchmark. 
    This superior performance stems from \thistool's typestate-guided context retrieval that enables sophisticated reasoning about cross-procedural complex memory error semantics. 
\end{answerBox}

\subsection{Comparison with LLM-based Approach (RQ2)}

\begin{figure}[t]
    \centering
    \includegraphics[width=\linewidth]{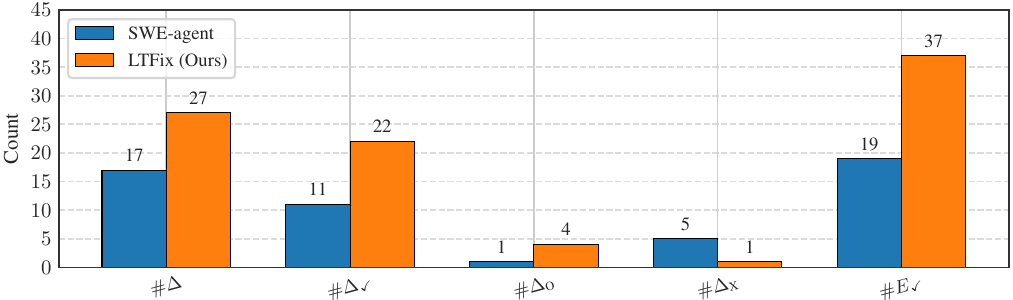}
    \vspace{-8mm}
    \caption{Comparison of fixing effectivenss between \thistool and SWE-agent~\cite{yang2024sweagent} in our dataset. }
    \vspace{-3mm}
    \label{fig:rq2-effectiveness}
\end{figure}

We conduct a comprehensive comparative analysis between \thistool and SWE-agent~\cite{yang2024sweagent}, an open-source state-of-the-art LLM-based program improvement tool. 

\runinhead*{Comparison Results.} As illustrated in Figure~\ref{fig:rq2-effectiveness}, \thistool demonstrates substantial performance improvements over SWE-agent across critical evaluation metrics. 
\thistool successfully repairs 37 memory errors, representing a 94.7\% increase compared to the 19 errors addressed by SWE-agent. 
Quantitatively, our approach generated 27 patches versus SWE-agent's 17 patches. The qualitative differential is even more pronounced: \thistool produces 22 correct patches—representing an 81.5\% accuracy rate—compared to SWE-agent's 10 correct patches (58.8\% accuracy). 
Furthermore, \thistool substantially reduces the generation of harmful patches that introduce new errors, with only 1 instance compared to SWE-agent's 5. 
These results empirically demonstrate that augmenting LLMs with typestate-guided context retrieval yields significantly higher repair precision than relying solely on the LLM's intrinsic reasoning capabilities.

SWE-agent is unable to repair the memory leak in Figure~\ref{fig:case-study}(a) because it lacks awareness of the database object's allocation/deallocation typestate transitions context. Without comprehensive visibility into the \code{db} structure's allocation pattern and complete memory lifecycle, SWE-agent cannot reliably determine whether \code{rec\_db\_destroy(db)} constitutes the appropriate replacement, as it must consider potential risks such as double-free vulnerabilities.
Similarly, SWE-agent fails to comprehend the intricate pointer alias relationships in Figure~\ref{fig:case-study}(c) without the execution context that captures memory references and state transitions. Moreover, it demonstrates inadequate performance when addressing the cyclic allocation scenario in Figure~\ref{fig:case-study}(d), primarily due to its inability to perform interprocedural memory lifecycle tracking and to identify the recurring allocation pattern that characterizes this particular vulnerability. These limitations underscore the fundamental advantage of typestate-guided context retrieval in providing the necessary semantic understanding for effective memory error repair, highlighting the qualitative difference between \thistool's targeted approach and the quantitative expansion of LLM context windows alone.

\begin{figure}[t]
    \centering
    \includegraphics[width=\linewidth]{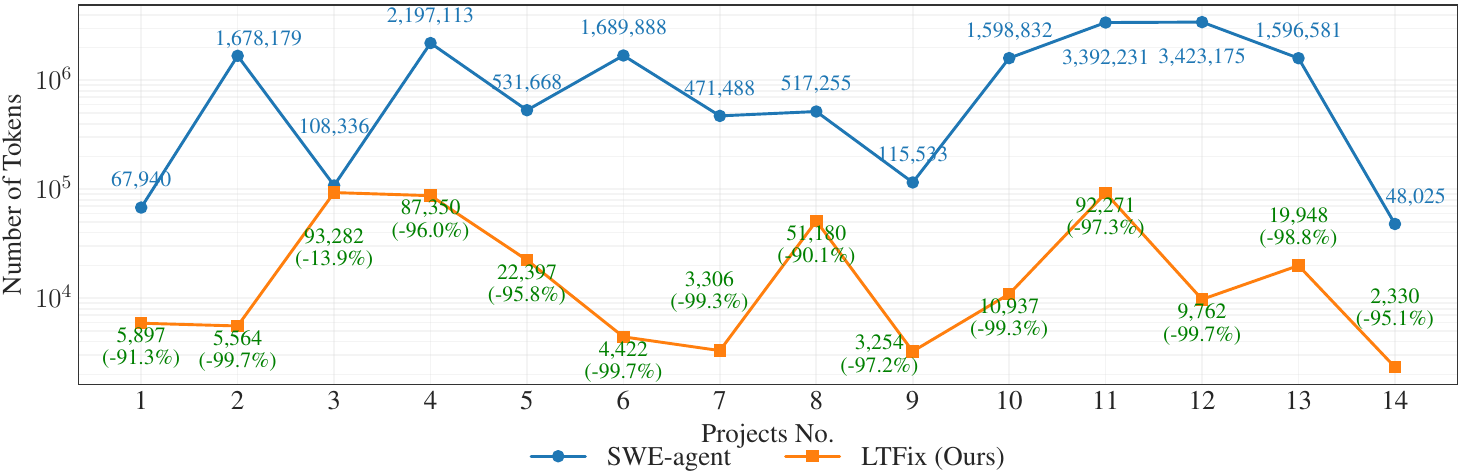}
    \vspace{-8mm}
    \caption{Comparison of the number of tokens consumed by \thistool and SWE-agent~\cite{yang2024sweagent} for memory error repair in our dataset. }
    \label{fig:rq2}
\end{figure}

\runinhead*{Token Consumption Analysis.} As illustrated in Figure~\ref{fig:rq2}, \thistool consumes significantly fewer tokens---approximately 42 times less than SWE-agent's total token usage. This substantial efficiency gap is consistently observed across all evaluated projects. The remarkable reduction in token consumption can be attributed to our typestate-guided context retrieval mechanism, which precisely extracts contextually relevant information necessary for generating accurate patches without incurring excessive computational overhead. In contrast, SWE-agent relies solely on the LLM's intrinsic reasoning capabilities without any specialized guidance for contextual prioritization.
These empirical results demonstrate that \thistool's strength extends beyond merely producing more patches; it generates substantially more accurate and less harmful repairs while addressing a greater number of errors---all with significantly higher token efficiency. This synergistic combination of repair quality and computational efficiency renders \thistool both more effective and more economical for practical deployment of LLM-based APR in real-world memory error repair scenarios.

\begin{answerBox}
\textbf{Answer to RQ2}: \thistool outperforms SWE-agent~\cite{yang2024sweagent} in terms of both repair accuracy and efficiency. Our approach generates substantially more correct patches, fewer harmful patches, and fixes more errors while consuming significantly fewer tokens, demonstrating superior performance in memory error repair.
\end{answerBox}

\subsection{Ablation Analysis (RQ3)}

Figure~\ref{fig:rq3} presents the ablation analysis result comparing \thistool with three variants (see \S\ref{sec:baselines}): \thistool-M (\thistool with the methods containing the error), \thistool-F (\thistool with the file containing the error) and \thistool-NT (\thistool without context trace).
We aim to understand the capability of the base LLM used in \thistool and the benefits brought by the typestate-guided context retrieval.

\runinhead*{Correct Patches.}
As shown in Figure~\ref{fig:rq3}(a), \thistool-F and \thistool-M, which provide file-level and method-level information respectively, demonstrate the lowest efficacy in generating correct patches, producing only 7 and 11 correct patches. 
\thistool-NT generates a higher number of correct patches than both baseline variants, achieving 15 correct patches. However, the complete \thistool system with intact typestate-guided contextual information significantly outperforms all variants with 22 correct patches, demonstrating a 47\% improvement over \thistool-NT and a 214\% improvement over \thistool-F.

\begin{figure}[t]
    \centering
    \includegraphics[width=\linewidth]{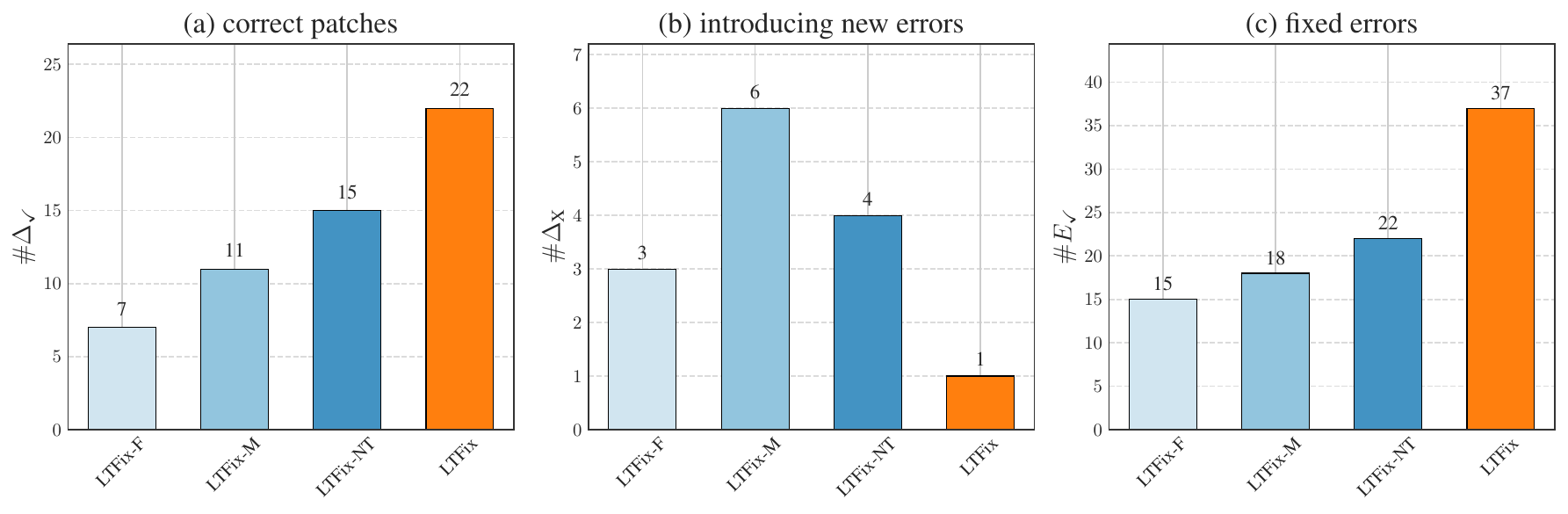}
    \vspace{-7mm}
    \caption{Ablation analysis result.}
    \vspace{-5mm}
    \label{fig:rq3}
\end{figure}

\runinhead*{Introducing New Errors.}
Figure~\ref{fig:rq3}(b) illustrates the number of new errors introduced by each variant. \thistool-F and \thistool-M introduce 3 and 6 new errors respectively, while \thistool-NT introduces 4 new errors.
In contrast, \thistool introduces only 1 new error, representing an 83\% reduction compared to the other variants. This significant decrease demonstrates that the integration of typestate-guided context retrieval is critical for preventing the introduction of new errors during the repair process.

\runinhead*{Fixed Errors.}
Figure~\ref{fig:rq3}(c) presents the total number of errors fixed by each variant. The baseline variants \thistool-F and \thistool-M fix 15 and 18 errors respectively, while the enhanced variant \thistool-NT shows incremental improvement by fixing 22 errors.
The complete \thistool system demonstrates superior performance by fixing 37 errors, which represents a 68\% improvement over \thistool-NT and a 147\% improvement over \thistool-F. This gap confirms that each component contributes significantly to the overall effectiveness of our approach, with typestate-guided context retrieval yielding synergistic benefits beyond what either component achieves independently.


To illustrate the critical importance of typestate-guided context retrieval, we examine the cases in Figure~\ref{fig:case-study} (c) and (d). Neither these variants can successfully repair these complex memory errors, as they lack the necessary execution context. Without typestate-guided context tracing, LLMs cannot identify critical semantic relationships—specifically, that pointer aliasing causes the double-free vulnerability in case (c), and that \code{note->font} undergoes multiple allocations without proper deallocation in case (d).
These examples demonstrate the fundamental necessity of context tracing in memory error repair. In Figure~\ref{fig:case-study} (c), context tracing enables the system to observe subtle changes in multiple pointers' states and track function call propagation across execution boundaries. Similarly, in Figure~\ref{fig:case-study} (d), it facilitates the tracking of \code{note->font}'s allocation state across diverse execution paths. In the absence of such context tracing, the LLM lacks comprehensive visibility into the object's complete lifecycle, rendering it impossible to determine that subsequent memory allocations occur without proper deallocation of previously allocated resources. This limitation fundamentally impedes the LLM's ability to generate semantically correct patches that address the underlying memory management deficiencies.

\runinhead*{Token Consumption Analysis.} 
To quantitatively assess the efficiency of typestate-guided context retrieval, we conduct a comparative analysis between our approach and the use of full context traces. The experimental results demonstrate that \thistool consumes only 411,900 tokens in total, whereas utilizing the complete context trace requires 22,950,414 tokens—a reduction factor exceeding 50$\times$. This substantial efficiency gain empirically validates the effectiveness of our targeted context retrieval strategy. By selectively extracting only the most semantically relevant information, \thistool maintains superior repair quality while dramatically reducing computational overhead, making it both more performant and more economical for practical deployment of LLM-based APR in real-world environments.

\begin{answerBox}
\textbf{Answer to RQ3}: The complete system's integration of all components in \thistool yields superior results across metrics, notably reducing new error introduction by enabling better understanding of error semantics, generating more precise repairs and consuming significantly less tokens.
\end{answerBox}
\section{Threats to Validity\label{sec:threats}}

\runinhead*{Dataset.} A potential threat to validity concerns the possible inclusion of our evaluated open-source projects and patches in the training dataset of the employed LLMs, which could introduce evaluation bias. Ideally, experiments would utilize new, previously unseen memory errors to completely eliminate this possibility. However, this limitation affects all LLM-based baselines equally in our comparative analysis, and our approach consistently demonstrates significant performance improvements over these baselines, suggesting that the core contributions of our typestate-guided context retrieval mechanism extend beyond any potential advantages from data exposure.

\runinhead*{LLM Selection.} Our evaluation primarily utilizes Claude 3.5 Sonnet, although empirical evidence suggests that comparable LLMs (e.g., GPT-4o) demonstrate similar performance on memory error APR tasks. While more sophisticated models might offer improvements, the focal point of our research is the enhancement of memory error APR through typestate-guided context retrieval, rather than a comparative assessment of performance across different LLM architectures. 

\runinhead*{Repair Scope.} Our approach does not aim to detect or repair all possible memory errors within a repository, but rather provides a targeted solution for memory errors that have been reproduced with a specific proof of concept. This may constrain the generalizability to memory errors beyond the provided proof of concept. However, to ensure methodological fairness in our comparative evaluation, all baseline approaches are provided with identical proof-of-concept demonstrations, and our approach consistently outperformed them under these controlled conditions.

\section{Related Work\label{sec:related}}

We review relevant literature across three primary domains: specialized techniques for automated memory error repair, the emerging integration of large language models in program repair frameworks, and the application of large language models for advanced program analysis. Through this examination, we position our approach within the broader research landscape while highlighting the limitations of existing methods when addressing complex memory management challenges.

\runinhead*{Automated Memory Error Repair.} 
Repairing memory errors is a complex task due to the non-local nature of memory management and its temporal properties. Various efforts have been proposed to address this issue~\cite{SAVER, ProveNFix, Footpatch, Memfix, nguyen2021apr, AddressWatcher, LeakFix,mechtaev2016angelix,semfix}. 
SemFix~\cite{semfix} and Angelix~\cite{mechtaev2016angelix} are general-purpose repair techniques that, while broadly applicable, demonstrate lower efficacy compared to specialized approaches tailored for specific error categories such as memory errors~\cite{SAVER} and null dereferences~\cite{vfix}.
AddressWatcher~\cite{AddressWatcher} can only fix memory leaks and cannot be applied to use-after-frees and double-frees. While MemFix~\cite{Memfix} is effective for small-scale programs, it struggles to scale up for larger applications and fails to generate patches that include conditional deallocation for safety checks. 
FootPatch~\cite{Footpatch} requires templated annotations at the bug locations and can inadvertently introduce double-free errors when addressing memory leaks. SAVER~\cite{SAVER} is more scalable for larger applications, but it does not take advantage of the intermediate bug information provided by Infer, which inhibits its effectiveness. \ProveNFix~\cite{ProveNFix} uses temporal property-based specifications, referred to as future conditions, to repair memory errors and other temporal bugs. However, a common limitation among these existing memory error APR tools is their dependence on manually crafted specifications.
In contrast, our tool can infer a correct fix without the need for explicitly defined rules. Moreover, compared to \thistool, existing solutions often struggle to fix inter-procedural multi-hunk bugs and have difficulty leveraging in-context repair beyond using primitive APIs.

\runinhead*{LLMs for Automated Program Repair.} Recent advancements in Automated Program Repair (APR)~\cite{APR_survey} have witnessed the emergence of LLM-based techniques, which can be categorized into two primary approaches: Open-Source-LLM-based \cite{OLLM_CodeBert-finetune21, OLLM_gamma, OLLM_Cure21, OLLM_selfAPR22, OLLM_RewardRepair22, NPR_LLMs23, OLLM_2024_Xia} and Closed-Source-LLM-based \cite{CLLM_CLPR, CLLM_SRepair, CLLM_inferfix, CLLM_CodeRover, CLLM_chatrepair, wei2023LLMAPR, OpenHands, yang2024sweagent, CodeAgent}. Open-Source-LLM-based methods typically necessitate substantial additional data for model fine-tuning. For instance, Mashhadi et al. \cite{OLLM_CodeBert-finetune21} curated over 600,000 samples to enhance Java single-statement bug repair capabilities. However, such extensive datasets are particularly challenging to assemble for memory error repairs due to their specialized nature.
Our approach aligns more closely with Closed-Source-LLM-based methodologies, which leverage the inherent capabilities of pre-trained LLMs as their foundation, subsequently augmenting them with external contextual information \cite{CLLM_chatrepair, CLLM_CLPR} or sophisticated decision-making chains~\cite{CLLM_CodeRover, CLLM_SRepair}. Recent research~\cite{OpenHands, yang2024sweagent, CodeAgent} has further evolved this paradigm by employing agent-based frameworks that enhance the repair process through interactive engagement with isolated computational environments.
While our approach shares the utilization of Closed-Source LLMs with existing methods, it fundamentally differs in both focus and implementation. Prior approaches predominantly target general bug fixing (mostly in Java and Python) but face limitations when addressing memory errors due to the extensive contextual information required. Our novel contribution lies in the development of typestate-guided context retrieval, which effectively compresses error-related contexts to lengths suitable for LLM processing while preserving critical semantic information necessary for accurate memory error repair.

\runinhead*{LLMs for Program Analysis.}
LLMs have demonstrated remarkable capabilities in reasoning about complex program semantics and performing sophisticated program analysis. For instance, Li et al.~\cite{li2024enhancing} effectively combine static analysis with LLMs to detect Use Before Initialization (UBI) bugs within the Linux kernel, exemplifying the potential of LLMs in understanding complex programming semantics. Similarly, Huang et al.~\cite{huang2024revealing} utilize the in-context learning capability of LLMs to elucidate program dependencies, highlighting their potential in clarifying intricate program structures. In another study, Wen et al.~\cite{wen2024enchanting} decompose programs and employ LLMs to synthesize specifications for automated program verification. Cheng et al.~\cite{cheng2024semanticenhancedindirectanalysislarge} propose a semantic-enhanced approach based on LLMs to improve the efficiency of indirect call analysis. Wang et al.~\cite{wang2025llmdfa} present an LLM-powered compilation-free and customizable dataflow analysis. While the focus of these studies differs from ours, the combination of program analysis for understanding program semantics and the in-context learning ability of LLMs has inspired us to leverage LLMs in the development of \thistool.

\section{Conclusion\label{sec:conclusion}}
The paper presents \thistool, a novel approach to automated memory error repair in C programs using Large Language Models (LLMs). By leveraging LLMs' extensive knowledge of code and natural language, guided by a finite typestate automaton and structured prompting, \thistool addresses the limitations of traditional automated memory error repair methods and previous deep learning approaches. Our tool demonstrates significant success in repairing real-world memory errors across large-scale open-source projects, outperforming existing state-of-the-art tools and even fixing three zero-day memory errors. This approach shows promise in advancing the field of automated program repair, particularly for complex memory-related errors in C programming.

\section*{Data Availability Statement}
This paper is currently under review. All implementation details and associated data are available to reviewers and will be made publicly available upon acceptance.


\bibliography{refs}



\end{document}